\documentclass[11pt,draftcls,onecolumn]{IEEEtranTCOM}

\usepackage{color}
\usepackage{amsfonts}
\usepackage{amsmath}
\usepackage{amsthm}
\usepackage{amssymb}
\usepackage{alltt}
\usepackage{epsfig}
\usepackage{mathrsfs}
\usepackage{graphicx,ifthen}
\usepackage{epstopdf}
\usepackage{subfigure}
\usepackage{esint}
\usepackage{bm}
\usepackage{enumerate}
\usepackage{cite}
\usepackage{float}
\usepackage{arydshln}

\newtheorem{proposition}{Proposition}
\newtheorem{lemma}{Lemma}
\newtheorem{definition}{Definition}

\newcommand{\mb}[1]{\mathbf{#1}}
\newcommand{\mc}[1]{\mathcal{#1}}
\newcommand{\mbb}[1]{\mathbb{#1}}

\begin{document}

\title{On the Mutual Information of Multi-RIS Assisted MIMO: From Operator-Valued Free Probability Aspect}

\date{} % Activate to display a given date or no date (if empty),
         % otherwise the current date is printed

\author{Zhong~Zheng, ~\IEEEmembership{Member,~IEEE,}
		Siqiang~Wang,
		Zesong~Fei,~\IEEEmembership{Senior Member,~IEEE,}
		Zhi~Sun,~\IEEEmembership{Senior Member,~IEEE,}
		Jinhong~Yuan,~\IEEEmembership{Fellow,~IEEE}
	\thanks{Z. Zheng, S. Wang, and Z. Fei are with the School of Information and Electronics, Beijing Institute of Technology, Beijing, China. Z. Sun is with the Department of Electronic Engineering, Tsinghua University, Beijing, China. J. Yuan is with the School of Electrical Engineering and Telecommunications, UNSW Sydney, Sydney, Australia.}
}

\markboth{}%
{Zheng \MakeLowercase{\textit{et al.}}: On the Mutual Information of Multi-RIS MIMO Channels}

\maketitle

\begin{abstract}
The reconfigurable intelligent surface (RIS) is useful to effectively improve the coverage and data rate of end-to-end communications. In contrast to the well-studied coverage-extension use case, in this paper, multiple RIS panels are introduced, aiming to enhance the data rate of multi-input multi-output (MIMO) channels in presence of insufficient scattering. Specifically, via the operator-valued free probability theory, the asymptotic mutual information of the large-dimensional RIS-assisted MIMO channel is obtained under the Rician fading with Weichselberger's correlation structure, in presence of both the direct and the reflected links. Although the mutual information of Rician MIMO channels scales linearly as the number of antennas and the signal-to-noise ratio (SNR) in decibels, numerical results show that it requires sufficiently large SNR, proportional to the Rician factor, in order to obtain the theoretically guaranteed linear improvement. This paper shows that the proposed multi-RIS deployment is especially effective to improve the mutual information of MIMO channels under the large Rician factor conditions. When the reflected links have similar arriving and departing angles across the RIS panels, a small number of RIS panels are sufficient to harness the spatial degree of freedom of the multi-RIS assisted MIMO channels.

\end{abstract}


\IEEEpeerreviewmaketitle

\section{Introduction}

In both the current and forthcoming generations of mobile communication systems, multi-input multi-output (MIMO) is one of the mainstream physical-layer techniques to improve the spectral efficiency and the reliability of the wireless communications~\cite{SaadMagazine2020}. In the favorable environments with rich scattering, MIMO is able to increase the achievable data rate linearly with the number of antennas~\cite{Telatar1999}. However, when the wireless systems operate in higher frequencies with larger bandwidth, such as the millimeter wave and terahertz bands, the radio signals are easily attenuated due to absorption and blockage. In this case, the MIMO channels typically have only a few dominating propagation paths and/or limited angular spread, which causes rank deficiency in the channel matrix that significantly degrades the MIMO channel capacity~\cite{ShinTIT2003}.

Recently, reconfigurable intelligent surface (RIS) has attracted substantial attentions and is foreseen to be an important component in the future communication systems~\cite{PanMagazine2021}. A typical RIS consists of a large number of low-power integrated electronic circuits, which can be programmed to modify the electromagnetic properties of the incoming radio waves in the desired frequency band~\cite{WuMagazine2020}, such as the phase and amplitude of the reflected signals from each programmable circuit. Therefore, by deploying some RIS panels in the environment, the signal's radiation pattern within the operating spectrum bands of the communication systems can be reconfigured to increase the number of independent paths with diversified angular spreads, thus increasing the rank of the MIMO channels. As an example, Fig.~\ref{figSystem} illustrates the transmissions between a base station (BS) and a user equipment (UE) in an urban canyon. In this scenario, without RIS deployment, the signals have to propagate through a scattering-limited area, where the direct propagation link $\mb{F}_0$ dominates the end-to-end channel, while other scattered/reflected components are severally attenuated by the building materials. In comparison, RIS panels are able to actively and effectively reflect the signals to increase the number of independent specular components, resulting in a total number of $K+1$ propagation links, including the direct link $\mb{F}_0$ and $K$ reflected links that consist of channels $\{\mb{F}_k\}_{1\le k\le K}$ between BS and RIS panels and channels $\{\mb{G}_k\}_{1\le k\le K}$ between RIS panels and UE.

\begin{figure}[t]
	\centerline{\includegraphics[width=0.6\columnwidth]{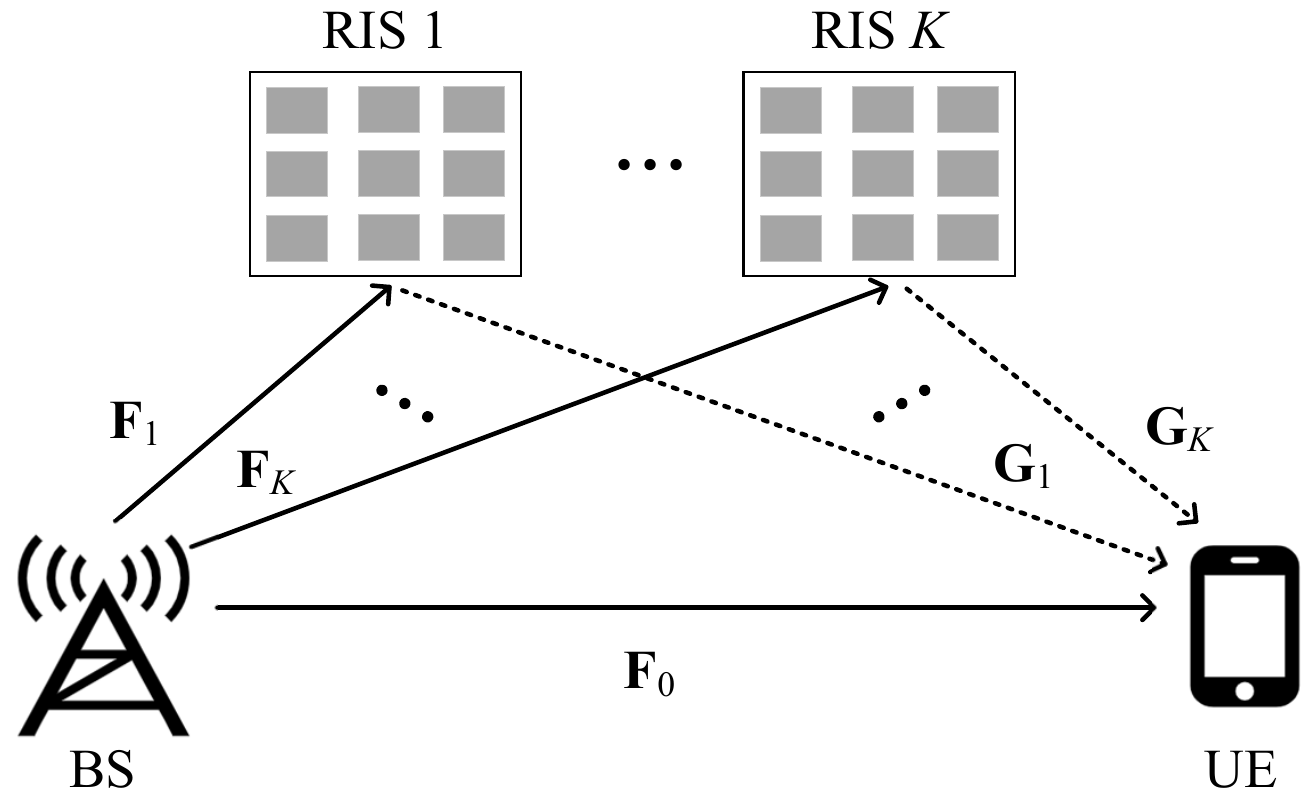}
	}
	\caption{Multi-RIS assisted MIMO communications.}
	\label{figSystem}
\end{figure}

There exist a number of studies focusing on the performance evaluation of the end-to-end communications assisted by a single or multiple RIS panels. When both the transmitter and receiver are equipped with a single antenna and a single RIS is deployed, the signal-to-noise-ratio (SNR) of such RIS-assisted single-input single-output (SISO) channel is proportional to the squared amplitude of the end-to-end effective channel. There are two typical theoretic frameworks to analyze the statistical properties of the SNR: One is based on the Meijer's G- and Fox's H-function systems~\cite{ZhangTVT2021}, which result in exact but rather complicated expressions. The other is to match the moments of the effective channel with classical random variables~\cite{YangWCL2020,SalhabWCL2021,DashCL2022,XieWCL2022}. In particular, when the direct link is blocked, the SNR distribution and the corresponding outage probability of the RIS-assisted communications are approximated by Gamma random variables, when the component channels are independently Rayleigh-faded~\cite{YangWCL2020}, independently Rician-faded~\cite{SalhabWCL2021}, correlated Rician-faded~\cite{DashCL2022}, and independently Nakagami-faded~\cite{XieWCL2022}, respectively. When both the direct and the reflected links exist, the Gamma-approximated SNR distribution and the finite block-length rate of the RIS-assisted channel are obtained in~\cite{HashemiTVT2021}, when all component channels are Rayleigh-faded.

When multiple RIS panels are deployed in the network, the RIS panels either work in exhaustive mode to jointly assist the end-to-end communications~\cite{YildirimTCOM2021}, or work in opportunistic mode, where only the RIS with maximum channel gains is selected~\cite{YangWCL2021}. In the former case, the end-to-end SNR is approximated by a Gaussian random variable due to the central limit theorem. The SNR and the average symbol error probability are derived for the phase shift keying signaling. In the latter case, the SNR of each reflected link corresponding to one RIS panel is approximated as a Gaussian random variable and the order-statistics is then obtained for the optimally selected RIS-assisting channel. In~\cite{DoTCOM2021}, a comprehensive performance comparison of those two operation modes is provided, under the scenario of multi-RIS assisted SISO channels assuming Nakagami-faded direct and reflected links.

In the case of MIMO systems, the available performance analysis of RIS-assisted communications is rather limited due to the challenge of understanding the statistical distribution of matrix-valued propagation channels, and results only exist for the single-RIS deployment. In~\cite{ShiTWC2022}, the authors consider the RIS-assisted MIMO communications, where the direct link is blocked and the reflected link is the concatenated Rayleigh-faded MIMO channels with single-sided correlation. The exact outage probability of such channel is derived by using the Mellin transform~\cite{Carter1977}. The result is expressed as the integration of product of multiple Meijer's G-functions, which is difficult to solve in practice. When the reflected link is the concatenated millimeter wave MIMO channels assuming Saleh-Valenzuela model~\cite{AyachTWC2014}, an upper bound of the ergodic achievable rate is derived in~\cite{LiTWCPress} using majorization theory and Jensen's inequality. In the RIS-assisted uplink multiple access channel, the asymptotic ergodic sum rate of the multi-user MIMO system is derived in~\cite{XuTCOM2021} by using the replica method, assuming that the reflected links are Rician-faded MIMO channels with Kronecker's correlation. In the same channel model as~\cite{XuTCOM2021}, the finite-SNR diversity-multiplexing tradeoff (DMT) of the RIS-assisted MIMO channel is analyzed in~\cite{ZhangJSTSP2022} by the martingale method. When both the direct and the reflected links exist, the asymptotic achievable rate of the single-RIS assisted MIMO channel is derived in~\cite{ZhangTWC2021} via replica method, assuming that all the component channels are Rician fading with Kronecker's correlation and all the channel dimensions grow to infinity.

Although the RIS-assisted MIMO communications have been investigated in~\cite{ShiTWC2022, LiTWCPress, XuTCOM2021, ZhangJSTSP2022, ZhangTWC2021}, the results therein are obtained for the single-RIS deployment under certain MIMO channel configurations. In contrast, this paper aims to provide the theoretic framework that analyzes the general multi-RIS assisted MIMO communications under arbitrary Rician fading with Weichselberger's correlation structure~\cite{WeichselbergerTWC06}. Such a model fits a wider range of realistic MIMO channels compared to the conventional Kronecker's correlation structure. Based on the above system settings, we first embed the component MIMO channel matrices into a large block matrix. Then, an operator-valued probability space over the algebra of the constructed block matrices is defined, where the operator-valued Cauchy transform is defined and is shown to be closely related to the classic Cauchy transform of the channel Gram matrix. The operator-valued Cauchy transform is then derived by leveraging the freeness over the defined probability space and the additive free convolution machinery. Based on the obtained Cauchy transform, the probability distribution of the eigenvalue of the channel Gram matrix as well as the mutual information of the multi-RIS assisted MIMO channel can be calculated, which avoids time-consuming Monte Carlo simulations. Numerical results show that in presence of strong line-of-sight conditions, although the mutual information could scale linearly as the number of antennas and the SNRs (in decibels), the SNR has to be sufficiently large in order to exhibit such linear scaling law. On the other hand, deploying additional RIS panels could effectively improve the channel's mutual information and thus, alleviate the SNR requirement.

The rest of this article is organized as follows. The signal model, the channel model, and the mutual information of the MIMO channel under consideration are introduced in Section~\ref{secModel}. In Section~\ref{secOperatorFree}, the operator-valued probability space is introduced and the main result of the Cauchy transform of the channel Gram matrix is given. Numerical simulation results on the spectral distribution and the mutual information of the MIMO channels are in Section~\ref{secResult}. Section~\ref{secConclude} concludes the main findings of this article.

\emph{Notations.} Throughout the paper, vectors and matrices are represented by lower-case and upper-case bold-face letters, respectively. The complex column vector with length $n$ is denoted as $\mbb{C}^n$. We use $\mc{CN}(\mb{0},\mb{A})$ to denote the zero-mean complex Gaussian vector with covariance matrix $\mb{A}$ and $\mb{I}_n$ is an $n\times n$ identity matrix. The superscript $(\cdot)^{\dag}$ denotes the matrix conjugate-transpose operation and $(\cdot)^{\mathrm{T}}$ is matrix transpose. We denote $\mathrm{Tr}(\mb{A})$ as the trace of $n\times n$ matrix $\mb{A}$. The notation $\mbb{E}[\cdot]$ denotes the expectation, and $\det(\cdot)$ denotes the matrix determinant.

\section{System Model}\label{secModel}

\subsection{Signal Model}
Consider a MIMO communication channel between a transmitter equipped with $T$ antennas and a receiver equipped with $R$ antennas. The transmissions are assisted by $K$ RIS panels, which reflect the impinging signals via their reflecting elements and each RIS panel is equipped with $L_k$ reflecting elements, $1\le k\le K$. For notational simplicity, we define $R = L_0$ and use these two symbols interchangeably. 

Denote the transmitted signal as $\mb{x}\in\mbb{C}^{T}$ and the additive noise at the receiver as $\mb{n}\in\mbb{C}^R$. The received signal $\mb{y}\in\mbb{C}^R$ is expressed as
\begin{align}
	\mb{y} = \left(\mb{F}_0 + \sum_{k=1}^{K} \sqrt{\rho_k}\mb{G}_k\mb{F}_k\right)\mb{x} + \mb{n},\label{eqy}
\end{align}
where the $R\times T$ matrix $\mb{F}_0$ denotes the direct channel between transmitter and receiver, the $L_k\times T$ matrix $\mb{F}_k$ denotes the channels between the transmitter and the $k$-th RIS panel, and $0<\rho_k\le 1$ denotes the relative channel gain of the $k$-th reflected channel via the $k$-th RIS, compared to the direct channel. The $R\times L_k$ matrix $\mb{G}_k$ denotes phase-shifted reflected channel between the $k$-th RIS and the receiver, modeled as
\begin{align}
	\mb{G}_k = \mb{R}_k\mb{\Theta}_k,
\end{align}
where the $R\times L_k$ matrix $\mb{R}_k$ denotes the channel coefficients, and the diagonal matrix $\mb{\Theta}_k = \mathrm{diag}\left(e^{i\phi_{k,1}},\ldots,e^{i\phi_{k,L_k}}\right)$ contains the phase-shifts of the reflecting elements, where $0\le\phi_{k,l}\le 2\pi$ denotes the phase-shift of the $l$-th element of the $k$-th RIS.  

We adopt the following assumptions on the signal and the channels:
\begin{enumerate}
	\item[(A1)]\label{assumption1} The signal $\mb{x}$ is Gaussian distributed with uniform power allocation, i.e., $\mb{x}\sim\mc{CN}(\mb{0}_T, P\mb{I}_T)$, where $P$ is the average power of the signals from each transmit antenna;
	\item[(A2)] The noise $\mb{n}$ is assumed to be a white Gaussian random vector with \emph{i.i.d.} zero-mean entries, i.e., $\mb{n}\sim\mc{CN}(\mb{0}_R, \sigma^2\mb{I}_R)$, where $\sigma^2$ denotes the variance of the noise;
	\item[(A3)] The channel coefficients $\{\mb{F}_k\}_{0\le k\le K}$ and $\{\mb{R}_k\}_{1\le k\le K}$ are block-faded, which keeps constant within the coherence time, while changing randomly and independently in the next coherence time. The phase shifts $\{\mb{\Theta}_k\}_{1\le k\le K}$ are   assumed to be fixed.
\end{enumerate}

Note that without the direct link $\mb{F}_0$, channel models similar to (\ref{eqy}) have been also studied in \cite{ShinTIT2003,ZhengTIT2017,MullerTIT2002a} for the keyhole channel, the Rayleigh-product channel, and the double-scattering channel, respectively, by using different theoretic techniques, which cannot be applied here.

\subsection{Channel Model}
In order to characterize the directivity and the spatial correlation of the channels between antenna arrays, we adopt the non-central Weichselberger's MIMO model for each component links~\cite{WeichselbergerTWC06}, such that
\begin{align}
	\mb{F}_k & = \overline{\mb{F}}_k + \widetilde{\mb{F}}_k = \overline{\mb{F}}_k + \mb{U}_k (\mb{M}_k \odot \mb{X}_k) \mb{V}_k^{\dagger},\quad 0\le k\le K,\label{eqFk}\\
	\mb{G}_k &= \overline{\mb{G}}_k + \widetilde{\mb{G}}_k = \overline{\mb{G}}_k + \frac{1}{\sqrt{r_k}}\mb{W}_k (\mb{N}_k\odot\mb{Y}_k) \mb{S}_k^{\dagger},\quad 1\le k\le K,\label{eqGk}
\end{align}
where $\overline{\mb{F}}_k$ and $\overline{\mb{G}}_k$ are the fixed specular components of $\mb{F}_k$ and $\mb{G}_k$, respectively. The random scattering components are captured by $\widetilde{\mb{F}}_k$ and $\widetilde{\mb{G}}_k$, where $\mb{U}_k$, $\mb{V}_k$, $\mb{W}_k$, and $\mb{S}_k$ are deterministic unitary matrices. The deterministic matrices $\mb{M}_k$ and $\mb{N}_k$ represent the variance profiles of $\widetilde{\mb{F}}_k$ and $\widetilde{\mb{G}}_k$, respectively, each having non-negative real elements. The $L_k\times T$ random matrix $\mb{X}_k$ and the $R\times L_k$ random matrix $\mb{Y}_k$ are \emph{i.i.d.} complex Gaussian distributed with entries having zero mean and variance $1/T$, i.e., $[\mb{X}_k]_{i,j}\sim\mc{CN}(0,1/T)$ and $[\mb{Y}_k]_{i,j}\sim\mc{CN}(0,1/T)$. We denote $r_k$ as the ratio between $L_k$ and $T$, i.e., $r_k = L_k/T$, $1\le k\le K$. The operator $\odot$ denotes the element-wise matrix multiplication. Note that the phase-shift matrix $\mb{\Theta}_k$, $1\le k\le K$, is also unitary and can be absorbed into the deterministic matrices $\overline{\mb{G}}_k$ and $\mb{S}_k$. The ratio between the power of fixed specular component and the random scattering component is defined as the Rician factor of the MIMO channel, i.e.,
\begin{align}
	\kappa_{k}^{(F)} = \frac{||\overline{\mb{F}}_k||_\mathrm{F}^2}{\mbb{E}[||\widetilde{\mb{F}}||_\mathrm{F}^2]},\  \mathrm{and}\ \kappa_{k}^{(G)} = \frac{||\overline{\mb{G}}_k||_\mathrm{F}^2}{\mbb{E}[||\widetilde{\mb{G}}||_\mathrm{F}^2]},
\end{align}
where $||\cdot||_\mathrm{F}$ denotes the Frobenius norm of a matrix.

For the correlated MIMO channel $\mb{G}_k$, the one-sided correlation function $\eta_k(\widetilde{\mb{C}}) = \mbb{E}[\widetilde{\mb{G}}_k^\dagger\widetilde{\mb{C}} \widetilde{\mb{G}}_k]$ parameterized by an Hermitian matrix $\widetilde{\mb{C}}$ is given by~\cite[Thm.~1]{WeichselbergerTWC06} as 
\begin{align}
	\eta_k(\widetilde{\mb{C}}) &= \mbb{E}[\widetilde{\mb{G}}_k^\dagger\widetilde{\mb{C}} \widetilde{\mb{G}}_k] = \frac{1}{L_k} {\mb{S}}_k\mb{\Pi}_k(\widetilde{\mb{C}}){\mb{S}}_k^\dagger,\quad 1\le k\le K,
\end{align}
where the $L_k\times L_k$ diagonal matrix $\mb{\Pi}_k(\widetilde{\mb{C}})$ contains the diagonal entries
\begin{align}
	\left[\mb{\Pi}_k(\widetilde{\mb{C}})\right]_{i,i} = \sum_{j = 1}^{R} \left([\mb{N}_k]_{j,i}\right)^2 \left[\mb{W}_k^\dagger \widetilde{\mb{C}} \mb{W}_k\right]_{j,j},\quad 1\le i\le L_k.
\end{align}
The other one-sided correlation function $\widetilde{\eta}_k({\mb{C}}_k) = \mbb{E}[\widetilde{\mb{G}}_k{\mb{C}}_k \widetilde{\mb{G}}_k^\dagger]$ parameterized by $\mb{C}_k$ is given by
\begin{align}
	\widetilde{\eta}_k({\mb{C}}_k) = \mbb{E}[\widetilde{\mb{G}}_k{\mb{C}}_k \widetilde{\mb{G}}_k^\dagger] = \frac{1}{L_k}\mb{W}_k \widetilde{\mb{\Pi}}_k({\mb{C}}_k) \mb{W}_k^\dagger,\quad 1\le k\le K,
\end{align}
where the $R\times R$ diagonal matrix $\widetilde{\mb{\Pi}}_k({\mb{C}}_k)$ contains the diagonal entries
\begin{align}
	\left[\widetilde{\mb{\Pi}}_k({\mb{C}}_k)\right]_{i,i} = \sum_{j=1}^{L_k} \left([\mb{N}_k]_{i,j}\right)^2 \left[ {\mb{S}}_k^\dagger {\mb{C}}_k {\mb{S}}_k \right]_{j,j}, \quad 1\le i\le R.
\end{align}
Similarly, for $0\le k\le K$, the two parameterized one-sided correlation functions of the matrix $\widetilde{\mb{F}}_k$ are given by:
\begin{align}
	\zeta_k(\mb{D}_k) &= \mbb{E}[\widetilde{\mb{F}}_k^\dagger\mb{D}_k\widetilde{\mb{F}}_k] = \frac{1}{T} \mb{V}_k\mb{\Sigma}_k(\mb{D}_k)\mb{V}_k^\dagger,\\
	\widetilde{\zeta}_k(\widetilde{\mb{D}}) &= \mbb{E}[\widetilde{\mb{F}}_k\widetilde{\mb{D}}\widetilde{\mb{F}}_k^\dagger] = \frac{1}{T}\mb{U}_k \widetilde{\mb{\Sigma}}_k(\widetilde{\mb{D}}) \mb{U}_k^\dagger,
\end{align}
where the $T\times T$ diagonal matrix $\mb{\Sigma}_k(\mb{D}_k)$ and the $L_k\times L_k$ diagonal matrix $\widetilde{\mb{\Sigma}}_k(\widetilde{\mb{D}})$ respectively contain the diagonal entries
\begin{align}
	\left[\mb{\Sigma}_k(\mb{D}_k)\right]_{i,i} &= \sum_{j = 1}^{L_k} \left([\mb{M}_k]_{j,i}\right)^2 \left[\mb{U}_k^\dagger \mb{D}_k \mb{U}_k\right]_{j,j},\quad 1\le i\le T,\\
	\left[\widetilde{\mb{\Sigma}}_k(\widetilde{\mb{D}})\right]_{i,i} &= \sum_{j=1}^{T} \left([\mb{M}_k]_{i,j}\right)^2 \left[ \mb{V}_k^\dagger \widetilde{\mb{D}} \mb{V}_k \right]_{j,j}, \quad 1\le i\le L_k.
\end{align}

In addition, since the channels $\{\mb{F}_k\}_{0\le k\le K}$ and $\{\mb{G}_k\}_{1\le k\le K}$ are spatially separated, channels correspond to different links are assumed to be independent.

\subsection{Mutual Information of Multi-RIS MIMO Channel}
Due to the assumptions~(A1)-(A3), the channel (\ref{eqy}) is a Gaussian MIMO channel and its mutual information is given by the well-known Telatar's formula~\cite{Telatar1999} as
\begin{align}
	I(\gamma) = \log\det\left(\mb{I}_R + \gamma \mb{H}\mb{H}^\dagger\right),\label{eqI_HH}
\end{align}
where $\gamma = P/\sigma^2$ is the average SNR, and the end-to-end channel $\mb{H}$ is given by
\begin{align}\label{eqH}
	\mb{H} = \mb{F}_0 + \sum_{k=1}^{K} \sqrt{\rho}_k\mb{G}_k\mb{F}_k.
\end{align}
The channel $\mb{H}$ can be factorized as the product of two matrices $\mb{G}$ and $\mb{F}$ as
\begin{align}
	\mb{H} &=  \mb{G}\mb{F} = \begin{bmatrix}
		\mb{I}_R & \sqrt{\rho_1} \mb{G}_1 & \ldots & \sqrt{\rho_K} \mb{G}_K
	\end{bmatrix} \begin{bmatrix}
		\mb{F}_0 \\ \mb{F}_1 \\ \vdots \\ \mb{F}_K
	\end{bmatrix}.% = \mb{G}(\overline{\mb{F}} + \widetilde{\mb{F}}),
\end{align}
Denoting $L = \sum_{k = 0}^{K} L_k$, $\mb{G} = \begin{bmatrix}
	\mb{I}_R & \sqrt{\rho_1} \mb{G}_1 & \ldots & \sqrt{\rho_K} \mb{G}_K
\end{bmatrix}$ is a $R\times L$ block matrix and $\mb{F} = \left[\mb{F}_0^\mathrm{T}, \ldots, \mb{F}_K^\mathrm{T}\right]^\mathrm{T}$ is a $L \times T$ block matrices. 

Letting $\mb{B} = \mb{H}\mb{H}^\dagger = \mb{G}\mb{F}\mb{F}^\dagger\mb{G}^\dagger$, the mutual information (\ref{eqI_HH}) can be rewritten as
\begin{align}
	I(\gamma) = R\  \mc{V}_\mb{B}(\gamma) = R \int_0^\infty \log(1+\gamma t)f_{\mb{B}}(t)\mathrm{d}t,\label{eqI_B}
\end{align}
where $\mc{V}_\mb{B}(x)$ is the Shannon transform of the matrix $\mb{B}$~\cite{DebbahBook}, and $f_{\mb{B}}(t)$ is the probability density function (PDF) of the eigenvalue of $\mb{B}$. Applying the relation between the Shannon transform and the corresponding Cauchy transform~\cite{DebbahBook}, the mutual information (\ref{eqI_B}) can be rewritten as
\begin{align}
	I(\gamma) = R\ \int_0^{\gamma} \left(\frac{1}{t} + \frac{1}{t^2}\mc{G}_{\mb{B}}\left(-\frac{1}{t}\right)\right)\mathrm{d}t,
\end{align}
where $\mc{G}_{\mb{B}}(z)$ is the Cauchy transform of $\mb{B}$ and is defined as
\begin{align}
	\mc{G}_{\mb{B}}(z) = \int_{0}^{\infty} \frac{1}{z-t} f_{\mb{B}}(t)\mathrm{d} t = \frac{1}{R}\mathrm{Tr}\circ \mbb{E}\left[\left(z\mb{I} - \mb{B}\right)^{-1}\right] = \tau_R\left(\left(z\mb{I} - \mb{B}\right)^{-1}\right). 
\end{align}
Here, $\tau_R(\mb{X})$ is the composite function $\frac{1}{R}\mathrm{Tr}\circ \mbb{E}[\mb{X}]$. Note that the PDF $f_{\mb{B}}(t)$ has an one-to-one mapping with the Cauchy transform $\mc{G}_{\mb{B}}(z)$ via the inverse transform
\begin{align}
	f_{\mb{B}}(t) = -\frac{1}{\pi} \lim_{\epsilon\rightarrow 0}\Im(\mc{G}_{\mb{B}}(t + i \epsilon)),\label{eqfB_invCauchy}
\end{align}
where $\Im(\cdot)$ denotes the imaginary part of the complex number. Therefore, the problem of finding the mutual information $I(\gamma)$ and the PDF $f_{\mb{B}}(t)$ are amount to finding the Cauchy transform $\mc{G}_{\mb{B}}(z)$ of product of matrices $\mb{B} = \mb{G}\mb{F}\mb{F}^\dagger\mb{G}^\dagger$. In the next section, we will resort to a linearization trick and the operator-valued free probability theory to derive the expression of~$\mc{G}_\mb{B}(z)$. 

\section{Asymptotic Eigenvalue Distribution via Operator-Valued Free Probability Theory}\label{secOperatorFree}

In the classic free probability theory, it is common to combine the Cauchy transform and the free multiplicative convolution to obtain the limiting spectral distribution of product of random matrices. For example, in \cite{MullerTIT2002}, the limiting spectral distribution of the concatenated MIMO channels of the form
\begin{align}\label{eqProdH}
	\left(\prod_{k=1}^{K}\mb{H}_k\right) \left(\prod_{k=1}^{K}\mb{H}_k\right)^\dagger
\end{align}
is derived, where $\mb{H}_k$ is $N_k\times N_{k-1}$ random matrix and has \emph{i.i.d.} zero-mean entries, unitarily invariant, and independent of each other. Such assumptions, in the language of free probability theory, is equivalent to requiring freeness among families of random variables as specified below. 

Let $\mc{A}$ be a unital algebra and $\mc{B} \subset \mc{A}$ be a unital subalgebra. For $\bm{\mc{H}}\in\mc{A}$, a linear map $\mbb{E}_{\mc{B}}[\bm{\mc{H}}]: \mc{A}\rightarrow\mc{B}$ is a $\mc{B}$-valued conditional expectation, if $\mbb{E}_{\mc{B}}[b] = b$ for all $b\in\mc{B}$, and $\mbb{E}_{\mc{B}}[b_1\bm{\mc{H}}b_2] = b_1 \mbb{E}_{\mc{B}}[\bm{\mc{H}}] b_2$ for all $\bm{\mc{H}}\in\mc{A}$ and $b_1, b_2\in\mc{B}$. Then, a $\mc{B}$-valued probability space is denoted as $(\mc{A}, \mbb{E}_{\mc{B}}, \mc{B})$, consisting of $\mc{B}\subset\mc{A}$ and the linear functional $\mbb{E}_{\mc{B}}$. In addition, let $\mc{A}_1,\ldots,\mc{A}_K$ be the subalgebras of $\mc{A}$ with $\mc{B}\subset\mc{A}_k$ for all $1\le k\le K$. We also let $\{\bm{\mc{H}}_k\in\mc{A}_k, 1\le k\le K\}$ denote a family of operator-valued random variables, which are free with amalgamation over $\mc{B}$ according to the following definition.
\begin{definition}\label{defConvention}
	Let $n$ be an arbitrary integer. The families of random variables $\{\bm{\mc{H}}_1,\ldots,\bm{\mc{H}}_K\}$ are  free with amalgamation over $\mc{B}$, if for every family of index $\{k_1,\ldots,k_n\}\subset\{1,\ldots,K\}$ with $k_1\neq k_2$, \ldots, $k_{n-1}\neq k_n$, and every family of polynomials $\{P_1,\ldots,P_n\}$ satisfying $\mbb{E}_\mc{B}[P_j(\bm{\mc{H}}_{k_j})] = 0$, $j\in\{1,\ldots,n\}$, we have $\mbb{E}_{\mc{B}}\left[\prod_{j=1}^{n} P_j(\bm{\mc{H}}_{k_j})\right] = 0$.
\end{definition}

In order to observe the freeness among families of random matrices $\{\mb{H}_k,\mb{H}_k^\dagger\}_{1\le k\le K}$ in (\ref{eqProdH}), we can construct the random variable $\bm{\mc{H}}_k$ as $\bm{\mc{H}}_k = \mb{H}_k\mb{H}_k^\dagger$. Let $\mc{C}$ denote the algebra of complex random variables. We define $\mc{A}_k = \mb{M}_{N_k}(\mc{C})$ as the algebra of $N_k\times N_k$ complex Hermitian matrices, $\mc{B} = \mc{C}$ and the linear functional $\mbb{E}_{\mc{B}}$ as $\mbb{E}_{\mc{B}} = \frac{1}{N_k}\mathrm{Tr}\circ\mbb{E}$. Then, as specified in Definition~\ref{defConvention}, the asymptotic freeness among $\left\{\bm{\mc{H}}_k\right\}_{1\le k\le K}$ over $\mc{C}$ has been established in some of the classic free probability theory, such as in~\cite{Capitaine2007}, which further enables one to apply free multiplicative convolution~\cite{MullerTIT2002} to obtain the limiting spectral distribution of the concatenated MIMO channels.

However, in the considered problem with $\mb{B} = \mb{G}\mb{F}\mb{F}^\dagger\mb{G}^\dagger$, both $\mb{G}$ and $\mb{F}$ are non-central and with non-trivial spatial correlations, and thus, are not free over $\mc{C}$ in the classic free probability aspect. More precisely, $\mb{G}\mb{G}^\dagger$ and $\mb{F}\mb{F}^\dagger$ are not free with respect to the linear functional $\tau_R = \frac{1}{R}\mathrm{Tr}\circ\mbb{E}$. Yet, as will be shown in the remaining of this section, via a linearization trick, the random matrix $\mb{B}$ of interest can be embedded into a larger block matrix, which can be then separated as the sum of a deterministic matrix and a random matrix. Instead of invoking the classic freeness over $\mc{C}$, we are able to elevate them as the operator-valued variables, which are shown to be asymptotically free in the operator-valued probability space. The limiting spectral distribution of their sum can be then obtained by using the operator-valued free additive convolution.

\subsection{Linearization Trick and Operator-Valued Probability Space}

Let $n$ denote $2L+R+T$ and $\mc{M} = \mb{M}_n(\mc{C})$ denote the algebra of $n\times n$ complex random matrices. Although the original formulation of $\mb{B}$ is in the form of product of two random matrices that are not free with respect to $\tau_R$, we could instead construct a block matrix $\mb{L}\in\mc{M}$, whose operator-valued Cauchy transform can be properly defined and is directly related to the conventional Cauchy transform of $\mb{B}$. 

By using the Anderson's linearization trick as described in~\cite[Prop. 3.4]{Speicher2013}, we can construct a block matrix $\mb{L}\in\mc{M}$ as follows
\begin{align}
	\mb{L} = \begin{bmatrix}
		\mb{L}^{(1,1)} & \mb{L}^{(1,2)}\\ 
		\mb{L}^{(2,1)} & \mb{L}^{(2,2)}
	\end{bmatrix} = \left[\begin{array}{c:ccc}
		\mb{0}_{R\times R} & \mb{0}_{R\times L} & \mb{0}_{R\times T} & \mb{G}\\\hdashline
		\mb{0}_{L\times R} & \mb{0}_{L\times L} & \mb{F} & -\mb{I}_{L}\\
		\mb{0}_{T\times R} & \mb{F}^\dagger & -\mb{I}_{T} & \mb{0}_{T\times L}\\
		\mb{G}^\dagger & -\mb{I}_{L} & \mb{0}_{L\times T} & \mb{0}_{L\times L}
	\end{array}\right],\label{eqBL}
\end{align}
where the matrix blocks $\left\{\mb{L}^{(i,j)}\right\}$ correspond to the partitions shown on the right-hand-side (RHS) of~(\ref{eqBL}). In addition, let us consider the sub-algebra $\mc{D}\subset\mc{M}$ as the $n\times n$ block diagonal matrix. For each $\mb{K}\in\mc{D}$, it is defined as
\begin{align}
	\mb{K} = \mathrm{blkdiag}\left(\widetilde{\mb{C}}, \mb{D}, \widetilde{\mb{D}}, \mb{C}\right),\label{eqK}
\end{align}
where $\widetilde{\mb{C}}$ is a $R\times R$ sub-matrix and $\widetilde{\mb{D}}$ is a $T\times T$ sub-matrix. The $L\times L$ block diagonal matrices $\mb{C}$ and $\mb{D}$ are defined as $\mb{C} = \mathrm{blkdiag}\left\{\mb{C}_0, \ldots, \mb{C}_K\right\}$ and $\mb{D} = \mathrm{blkdiag}\left\{\mb{D}_0, \ldots, \mb{D}_K\right\}$, respectively, where $\mb{C}_k$ and $\mb{D}_k$ are $L_k\times L_k$ sub-matrices. In (\ref{eqK}), all the involved sub-matrices $\widetilde{\mb{C}}$, $\widetilde{\mb{D}}$, $\left\{\mb{C}_k\right\}_{0\le k\le K}$, and $\left\{\mb{D}_k\right\}_
{0\le k\le K}$ are Hermitian matrices. 

For $\mb{X}\in\mc{M}$, we define $\mb{X}_{\widetilde{\mb{C}}}$, $\mb{X}_{\widetilde{\mb{D}}}$, $\left\{\mb{X}_{\mb{C}_k}\right\}_{0\le k\le K}$, and $\left\{\mb{X}_{\mb{D}_k}\right\}_{0\le k\le K}$ as the sub-blocks of $\mb{X}$, corresponding to the same diagonal sub-blocks $\widetilde{\mb{C}}$, $\widetilde{\mb{D}}$, $\left\{\mb{C}_k\right\}_{0\le k\le K}$, and $\left\{\mb{D}_k\right\}_{0\le k\le K}$ in the matrix $\mb{K}$. Then, we define the linear functional $\tau_\mc{D}: \mc{M}\rightarrow \mc{D}$ as
\begin{align}
	\tau_\mc{D}(\mb{X}) = \mathrm{id}\circ\mbb{E}_\mc{D}\left[\mb{X}\right],
\end{align}
where $\mathrm{id}$ denotes the identity operator on a Hilbert space and the expectation $\mbb{E}_\mc{D}\left[\mb{X}\right]$ is defined as
\begin{align}
	\mbb{E}_\mc{D}\left[\mb{X}\right] = \left[\begin{array}{c:c:c:c}
		\mbb{E}[\mb{X}_{\widetilde{\mb{C}}}] & & & \\\hdashline
		& \mbb{E}[\mb{X}_{\mb{D}}] & & \\\hdashline
		& & \mbb{E}[\mb{X}_{\widetilde{\mb{D}}}] & \\\hdashline
		& & &  \mbb{E}[\mb{X}_{\mb{C}}]
	\end{array}\right],\label{eqEDX}
\end{align}
and $\mbb{E}[\mb{X}_{\mb{C}}] = \mathrm{blkdiag}\left\{\mbb{E}[\mb{X}_{\mb{C}_0}], \ldots, \mbb{E}[\mb{X}_{\mb{C}_K}]\right\}$, $\mbb{E}[\mb{X}_{\mb{D}}] = \mathrm{blkdiag}\left\{\mbb{E}[\mb{X}_{\mb{D}_0}], \ldots, \mbb{E}[\mb{X}_{\mb{D}_K}]\right\}$. Then, we can define an operator-valued probability space $(\mc{M}, \tau_{\mc{D}}, \mc{D})$. For the $\mc{M}$-valued random variable $\mb{L}\in (\mc{M},\tau_\mc{D},\mc{D})$, its $\mc{D}$-valued Cauchy transform is defined as
\begin{align}
	\mc{G}_{\mb{L}}^{\mc{D}}(\mb{\Lambda}(z)) = \mathrm{id}\circ\mbb{E}_{\mc{D}}\left[\left(\mb{\Lambda}(z) - \mb{L}\right)^{-1}\right] = \tau_{\mc{D}}\left(\left(\mb{\Lambda}(z) - \mb{L}\right)^{-1}\right),\label{eqGBL0}
\end{align}
where $\mb{\Lambda}(z)\in\mc{M}$ denotes the $n\times n$ diagonal matrix as
\begin{align}
	\mb{\Lambda}(z) = \begin{bmatrix}
		z\mb{I}_R & \mb{0}_{R\times (2L+T)}\\
		\mb{0}_{(2L+T)\times R} & \mb{0}_{(2L+T)\times (2L+T)} 
	\end{bmatrix}.\label{eqLambdaz}
\end{align}

By substituting (\ref{eqBL}) and (\ref{eqLambdaz}) into (\ref{eqGBL0}) and invoking Lemma~\ref{lemmaBlockInv2} in the Appendix~\ref{appLemma}, we obtain
\begin{align}
	\mc{G}_{\mb{L}}^{\mc{D}}(\mb{\Lambda}(z)) &= \mathrm{id}\circ \mbb{E}_{\mc{D}}\begin{bmatrix}
		\left(z \mb{I}_{R} + \mb{L}^{(1,2)}\left(\mb{L}^{(2,2)}\right)^{-1}\mb{L}^{(2,1)}\right)^{-1} & -\mb{L}^{(1,2)}\left(z\mb{L}^{(2,2)} + \mb{L}^{(2,1)}\mb{L}^{(1,2)}\right)^{-1} \\
		 -\left(z\mb{L}^{(2,2)} + \mb{L}^{(2,1)}\mb{L}^{(1,2)}\right)^{-1}\mb{L}^{(2,1)} & -\left(\mb{L}^{(2,2)} + z^{-1}\mb{L}^{(2,1)}\mb{L}^{(1,2)}\right)^{-1}
	\end{bmatrix}.\label{eqBL1}
\end{align}
In particular, the upper-left block of (\ref{eqBL1}) can be explicitly written as
\begin{align}
	\left(z \mb{I}_{R} + \mb{L}^{(1,2)}\left(\mb{L}^{(2,2)}\right)^{-1}\mb{L}^{(2,1)}\right)^{-1} = \left(z\mb{I}_R - \mb{G}\mb{F}\mb{F}^\dagger\mb{G}^\dagger\right)^{-1}.
\end{align}
Therefore, the Cauchy transform of $\mb{B}$ over $\mc{C}$ is related to the $\mc{D}$-valued Cauchy transform of $\mb{L}$ as
\begin{align}
	\mc{G}_\mb{B}(z) = \frac{1}{R}\mathrm{Tr}\left(\left\{\mc{G}_{\mb{L}}^{\mc{D}}(\mb{\Lambda}(z))\right\}^{(1,1)}\right),\label{eqGB_GL}
\end{align}
where $\{\cdot\}^{(1,1)}$ denotes the upper-left $R\times R$ matrix block. 

\subsection{Operator-Valued Free Additive Convolution}

Let us introduce the following notations:
\begin{align}
	\overline{\mb{G}} &= \begin{bmatrix}
		\mb{I}_R & \sqrt{\rho_1}\overline{\mb{G}}_1 & \ldots & \sqrt{\rho_K}\overline{\mb{G}}_K
	\end{bmatrix},\\
	\widetilde{\mb{G}} &= \begin{bmatrix}
		\mb{0}_R & \sqrt{\rho_1}\widetilde{\mb{G}}_1 & \ldots & \sqrt{\rho_K}\widetilde{\mb{G}}_K
	\end{bmatrix},\\
	\overline{\mb{F}} &= \begin{bmatrix}
		\overline{\mb{F}}_0^\mathrm{T} & \ldots & \overline{\mb{F}}_K^\mathrm{T}
	\end{bmatrix}^\mathrm{T},\\
	\widetilde{\mb{F}} &= \begin{bmatrix}
		\widetilde{\mb{F}}_0^\mathrm{T} & \ldots & \widetilde{\mb{F}}_K^\mathrm{T}
	\end{bmatrix}^\mathrm{T}.
\end{align}
The linearization matrix $\mb{L}$ can be further expressed as
\begin{align}
	\mb{L} = \overline{\mb{L}} + \widetilde{\mb{L}},
\end{align}
where $\overline{\mb{L}}$ and $\widetilde{\mb{L}}$ contain the deterministic and the random parts of $\mb{L}$, respectively, and are given as follows:
\begin{align}
	\overline{\mb{L}} = \left[\begin{array}{c:c:c:c}
		&  &  & \overline{\mb{G}} \\\hdashline
		&  & \overline{\mb{F}} & -\mb{I}_{L}\\\hdashline
		& \overline{\mb{F}}^\dagger & -\mb{I}_{T} & \\\hdashline
		\overline{\mb{G}}^\dagger & -\mb{I}_{L} &  & 
	\end{array}\right],\label{eqBL_bar}
\end{align}
\begin{align}
	\widetilde{\mb{L}} = \left[\begin{array}{c:c:c:c}
		&  &  & \widetilde{\mb{G}} \\\hdashline
		&  & \widetilde{\mb{F}} & \\\hdashline
		& \widetilde{\mb{F}}^\dagger &  & \\\hdashline
		\widetilde{\mb{G}}^\dagger &  &  & 
	\end{array}\right],
\end{align}
where we omit the all-zero matrix blocks.

The advantage of working with $\mb{L}$ as well as its $\mc{D}$-valued Cauchy transform $\mc{G}_{\mb{L}}^{\mc{D}}$ is that the elements of $\widetilde{\mb{L}}$ are monomials of $\widetilde{\mb{G}}$, $\widetilde{\mb{G}}^\dagger$, $\widetilde{\mb{F}}$, and $\widetilde{\mb{F}}^\dagger$, which are decoupled from each other. This is in contrast to the Cauchy transform of $\mb{B}$ over $\mc{C}$, where the random variables are mixed together. Then, following similar steps as in~\cite{LuTIT2016}, $\widetilde{\mb{L}}$ is shown to be an operator-valued semicircular variable and is free from the deterministic matrix $\overline{\mb{L}}$ over $\mc{D}$. Thus, the limiting spectral distribution of $\mb{L}$ can be determined by the operator-valued free additive convolution of $\overline{\mb{L}}$ and $\widetilde{\mb{L}}$, over the sub-algebra $\mc{D}$, which are summarized in the following propositions.

\begin{proposition}\label{prop_freeness}
	The random variable $\widetilde{\mb{L}}$ is semicircular and is free from $\overline{\mb{L}}$ over $\mc{D}$.
\end{proposition}
\begin{proof}
The proof of Proposition~\ref{prop_freeness} is given in Appendix~\ref{appx_freeness}.
\end{proof}

Due to Proposition~\ref{prop_freeness}, the operator-valued Cauchy transform of $\mb{L}$ in (\ref{eqGB_GL}) can be calculated as the free additive convolution between $\widetilde{\mb{L}}$ and $\overline{\mb{L}}$, by using a subordination formula~\cite{Speicher2013} as follows:
\begin{align}
	\mc{G}_{\mb{L}}^{\mc{D}}(\mb{\Lambda}(z)) &= \mc{G}_{\overline{\mb{L}}}^{\mc{D}}\left(\mb{\Lambda}(z) - \mc{R}_{\widetilde{\mb{L}}}^{\mc{D}}\left(\mc{G}_{\mb{L}}^{\mc{D}}(\mb{\Lambda}(z))\right)\right)\nonumber\\
	&= \mbb{E}_{\mc{D}}\left[\left(\mb{\Lambda}(z) - \mc{R}_{\widetilde{\mb{L}}}^{\mc{D}}\left(\mc{G}_{\mb{L}}^{\mc{D}}(\mb{\Lambda}(z))\right) - \overline{\mb{L}}\right)^{-1}\right],\label{eqGL_sub}
\end{align}
where $\mc{R}_{\widetilde{\mb{L}}}^{\mc{D}}\left(\cdot\right)$ denotes the operator-valued $R$-transform of $\mb{L}$ over $\mc{D}$. Then, $\mc{G}_\mb{B}(z)$ can be determined by the following proposition.

\begin{proposition}\label{prop_cauchyB}
	The Cauchy transform of $\mb{B}$, with $z\in\mbb{C}^+$, is given by
	\begin{align}\label{eqGB}
		\mc{G}_{\mb{B}}(z) = \frac{1}{R}\mathrm{Tr}\left[\left(\widetilde{\mb{\Psi}}(z) - \overline{\mb{G}} \mb{\Xi}(z)^{-1} \overline{\mb{G}}^\dagger\right)^{-1}\right],
	\end{align}
	where
	\begin{align}
		\mb{\Xi}(z) &= \mb{\Psi}(z) - \left(\widetilde{\mb{\Phi}}(z) - \overline{\mb{F}}\mb{\Phi}(z)^{-1}\overline{\mb{F}}^\dagger\right)^{-1}.
	\end{align}
The matrix-valued function $\widetilde{\mb{\Psi}}(z)$, $\mb{\Psi}(z)$, $\widetilde{\mb{\Phi}}(z)$, $\mb{\Phi}(z)$ satisfy the following fixed-point equations
\begin{align}
	\widetilde{\mb{\Psi}}(z) &= z\mb{I}_R - \sum_{k=1}^K \widetilde{\eta}_k(\mc{G}_{{\mb{C}}_k}(z)),\label{eqPsit}\\
	\mb{\Psi}(z) &= \mathrm{blkdiag}\left\{ \mb{0}_R,\  -\eta_1(\mc{G}_{\widetilde{\mb{C}}}(z)),\  \ldots,\  -\eta_K(\mc{G}_{\widetilde{\mb{C}}}(z))\right\},\label{eqPsi}\\
	\widetilde{\mb{\Phi}}(z) &= \mathrm{blkdiag}\left\{ - \widetilde{\zeta}_0(\mc{G}_{\widetilde{\mb{D}}}(z)),\ - \widetilde{\zeta}_1(\mc{G}_{\widetilde{\mb{D}}}(z)),\ \ldots,\   - \widetilde{\zeta}_K(\mc{G}_{\widetilde{\mb{D}}}(z)) \right\},\label{eqPhit}\\
	\mb{\Phi}(z) &= \mb{I}_T-\sum_{k=0}^K \zeta_k(\mc{G}_{\mb{D}_k}(z)),\label{eqPhi}
\end{align}
where $\mathrm{blkdiag}\left\{\mb{A}_1,\ldots,\mb{A}_n\right\}$ constructs a block diagonal matrix with square matrices $\mb{A}_1,\ldots,\mb{A}_n$ being the diagonal blocks, and $\mc{G}_{\widetilde{\mb{C}}}(z)$, $\mc{G}_{\mb{C}_k}(z)$, $\mc{G}_{\widetilde{\mb{D}}}(z)$, $\mc{G}_{\mb{D}_k}(z)$ are given by
\begin{align}
	\mc{G}_{\widetilde{\mb{C}}}(z) &= \left(\widetilde{\mb{\Psi}}(z) - \overline{\mb{G}} \mb{\Xi}(z)^{-1}\overline{\mb{G}}^\dagger \right)^{-1},\label{eqGCt}\\
	\mc{G}_{\mb{C}_k}(z) &= \left\{\left(\mb{\Psi}(z) - \overline{\mb{G}}^\dagger \widetilde{\mb{\Psi}}(z)^{-1}\overline{\mb{G}} - \left( \widetilde{\mb{\Phi}}(z) - \overline{\mb{F}}\mb{\Phi}(z)^{-1}\overline{\mb{F}}^\dagger \right)^{-1} \right)^{-1}\right\}_{k+1}, \quad 1\le k\le K,\label{eqGCk}\\
	\mc{G}_{\widetilde{\mb{D}}}(z) &= \left(\mb{\Phi}(z) - \overline{\mb{F}}^\dagger \left(\widetilde{\mb{\Phi}}(z) - \left(\mb{\Psi}(z) - \overline{\mb{G}}^\dagger\widetilde{\mb{\Psi}}(z)^{-1}\overline{\mb{G}} \right)^{-1}\right)^{-1} \overline{\mb{F}} \right)^{-1},\label{eqGDt}\\
	\mc{G}_{\mb{D}_k}(z) &= \left\{ \left(\widetilde{\mb{\Phi}}(z) - \overline{\mb{F}}\mb{\Phi}(z)^{-1}\overline{\mb{F}}^\dagger - \left(\mb{\Psi}(z) - \overline{\mb{G}}^\dagger \widetilde{\mb{\Psi}}(z)^{-1}\overline{\mb{G}}\right)^{-1}\right)^{-1} \right\}_{k+1}, \quad 0\le k\le K.\label{eqGDk}
\end{align}
The notation $\{\mb{A}\}_{k+1}$ with $n\times n$ matrix $\mb{A}$ denotes the $(k+1)$-th diagonal matrix block containing entries from $\sum_{i=0}^{k-1}L_i+1$ to $\sum_{i=0}^{k} L_i$ rows and columns of $\mb{A}$.
\end{proposition}
\begin{proof}
	The proof of Proposition~\ref{prop_cauchyB} is given in Appendix~\ref{appx_cauchyB}.
\end{proof}

As indicated by Proposition~\ref{prop_cauchyB}, the Cauchy transform $\mc{G}_{\mb{B}}(z)$ as well as the matrix-valued functions $\widetilde{\mb{\Psi}}(z)$, $\mb{\Psi}(z)$, $\widetilde{\mb{\Phi}}(z)$, $\mb{\Phi}(z)$ can be determined by solving the fixed-point equations. The numerical value of $\mc{G}_{\mb{B}}(z)$ can be obtained by iterating the set of equations (\ref{eqPsit})-(\ref{eqPhi}) and (\ref{eqGCt})-(\ref{eqGDk}).

\section{Numerical Results}\label{secResult}

In this section, numerical simulations are conducted to study the spectral distribution of the RIS-assisted MIMO channel as well as its mutual information. In particular, we examine the impacts of the number of RIS panels, the number of antennas at the transceivers, and the Rician factor of the propagation channels on the mutual information. The mutual information $I(\gamma)$ and the eigenvalue PDF $f_{\mb{B}}(t)$ are calculated by (\ref{eqI_B}) and (\ref{eqfB_invCauchy}), respectively, where the involved Cauchy transform $\mc{G}_{\mb{B}}(z)$ is given in Proposition~\ref{prop_cauchyB}. In each simulation case, the MIMO system without RIS deployment is included for comparison, i.e., $K = 0$, where the eigenvalue PDF and the Cauchy transform can be calculated by using existing result from~\cite[Thm. 2]{LuTIT2016}. Each simulation curve is obtained by averaging over $10^6$ independent channel realizations.

In the simulations, the antenna elements of the transceivers and the reflecting elements of the RIS panels are arranged as the uniform planar arrays (UPAs). Denote  $T = T^{(H)}\times T^{(V)}$, $R = R^{(H)}\times R^{(V)}$, and $L_k = L_k^{(H)}\times L_k^{(V)}$, where the numbers with the superscripts $H$ and $V$ represent the numbers of elements aligned in the horizontal and vertical dimensions, respectively. The specular component of each channel is the line-of-sight propagation component between two uniform planar arrays (UPAs), i.e., 
\begin{align}
	\overline{\mb{F}}_k &= \mb{a}\left(\varphi_k^{(F)},\nu_k^{(F)}, L_k^{(H)}, L_k^{(V)}\right) \mb{a}^\dagger\left(\theta_k^{(F)},\phi_k^{(F)}, T^{(H)}, T^{(V)}\right),\quad 0\le k\le K,\label{eqF_UPA}\\
	\overline{\mb{G}}_k &= \mb{a}\left(\varphi_k^{(G)},\nu_k^{(G)}, R^{(H)}, R^{(V)}\right) \mb{a}^\dagger\left(\theta_k^{(G)},\phi_k^{(G)}, L_k^{(H)}, L_k^{(V)}\right),\quad 1\le k\le K,\label{eqG_UPA}
\end{align}
where $\theta_k^{(i)}$ and $\phi_k^{(i)}$ are the azimuth and elevation angles of the $k$-th departing UPA, while $\varphi_k^{(i)}$ and $\nu_k^{(i)}$ are the azimuth and elevation angles of the $k$-th arriving UPA, $i\in\{F,G\}$. The function $\mb{a}(\cdot)$ denotes the steering vector of an $M\times N$ UPA and is defined as 
\begin{align}
	\mb{a}(\alpha,\beta, M, N) = \left[1,\ldots, e^{i\pi (n\sin(\alpha)\sin(\beta) + m \cos(\beta))},\ldots,e^{i\pi ((N-1)\sin(\alpha)\sin(\beta) + (M-1) \cos(\beta))}\right]^\mathrm{T},
\end{align}
where $0\le m\le M-1$ and $0\le n\le N-1$.

Fig.~\ref{figPDF} shows the empirical and asymptotic eigenvalue PDF of the RIS-assisted MIMO channels $\mb{H}\mb{H}^\dagger$, assuming the number of RIS panels is $K = 0$, $1$, $2$, and $4$, respectively. In all the cases, the numbers of transmit and receive antennas are set to $T = R = 64$, and the number of reflecting elements in each RIS panel is set to 144. The channel statistics, such as $\overline{\mb{F}}_k$, $\mb{U}_k$, $\mb{V}_k$, $\mb{M}_k$ in (\ref{eqFk}), and $\overline{\mb{G}}_k$, $\mb{W}_k$, $\mb{S}_k$, $\mb{N}_k$ in (\ref{eqGk}) are randomly generated but fix for the Monte Carlo simulations. The numerical results show that the asymptotic PDF calculated by (\ref{eqfB_invCauchy}) provides an excellent approximation to the simulated PDF for all the considered parameter configurations. By increasing the number of deployed RIS panels, it is possible to increase the maximum eigenvalue, therefore, improve amplitude of the eigen-channels.

\begin{figure*}[t!]
	\centering
	\subfigure[$K = 0$]{\includegraphics[width=2.8in]{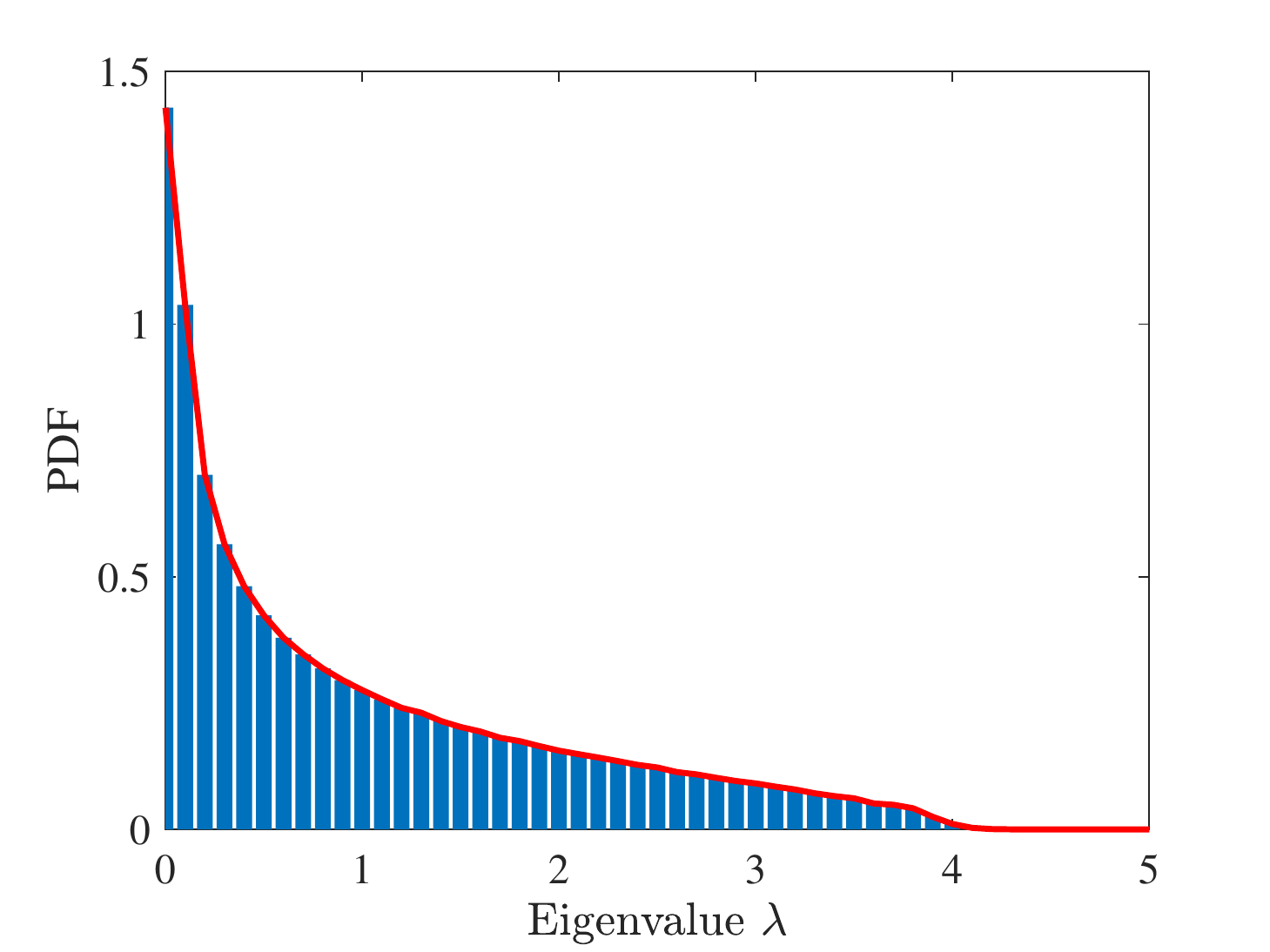}
	}
	\subfigure[$K = 1$]{\includegraphics[width=2.8in]{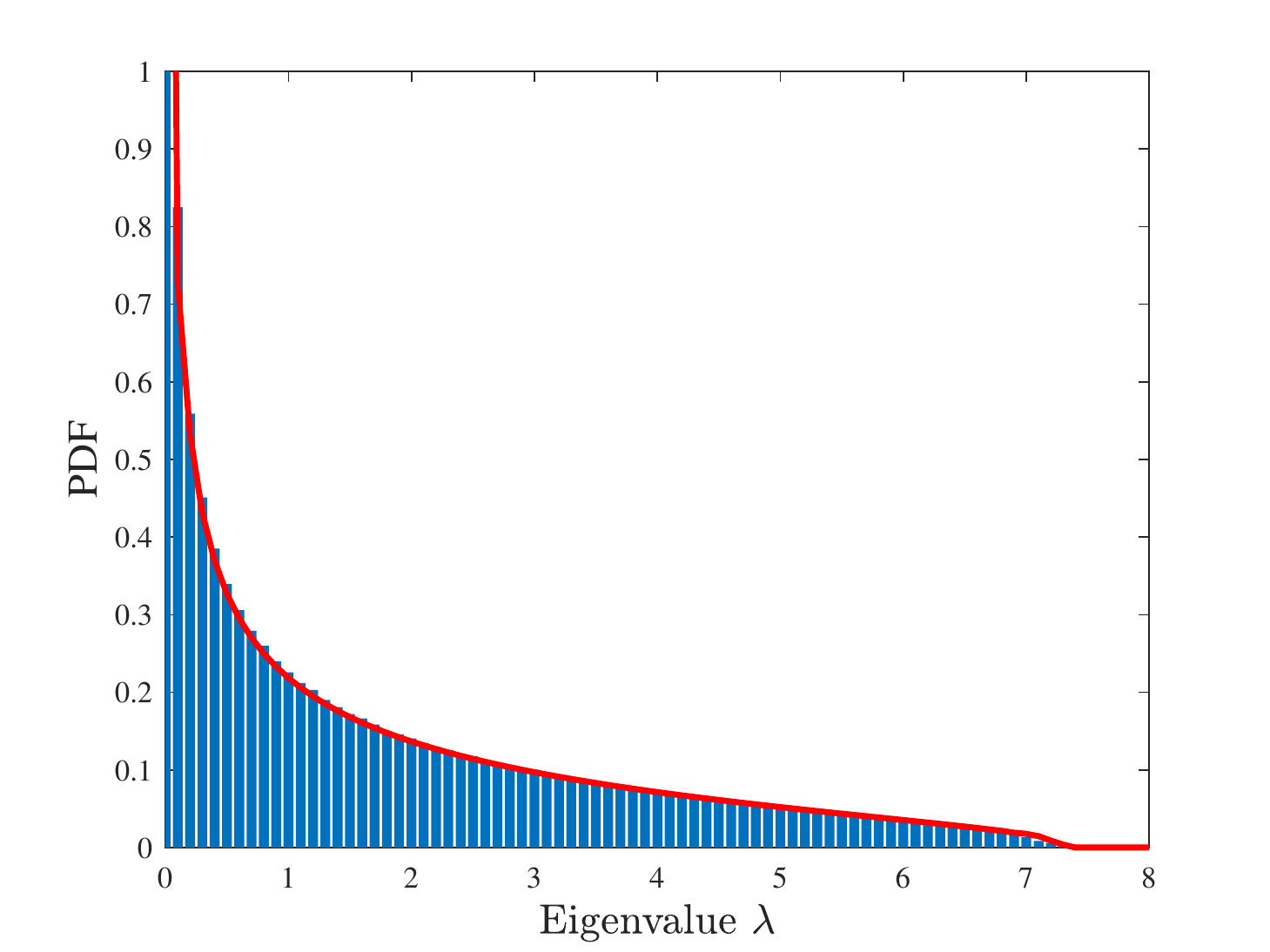}
	}\\
	\subfigure[$K = 2$]{\includegraphics[width=2.8in]{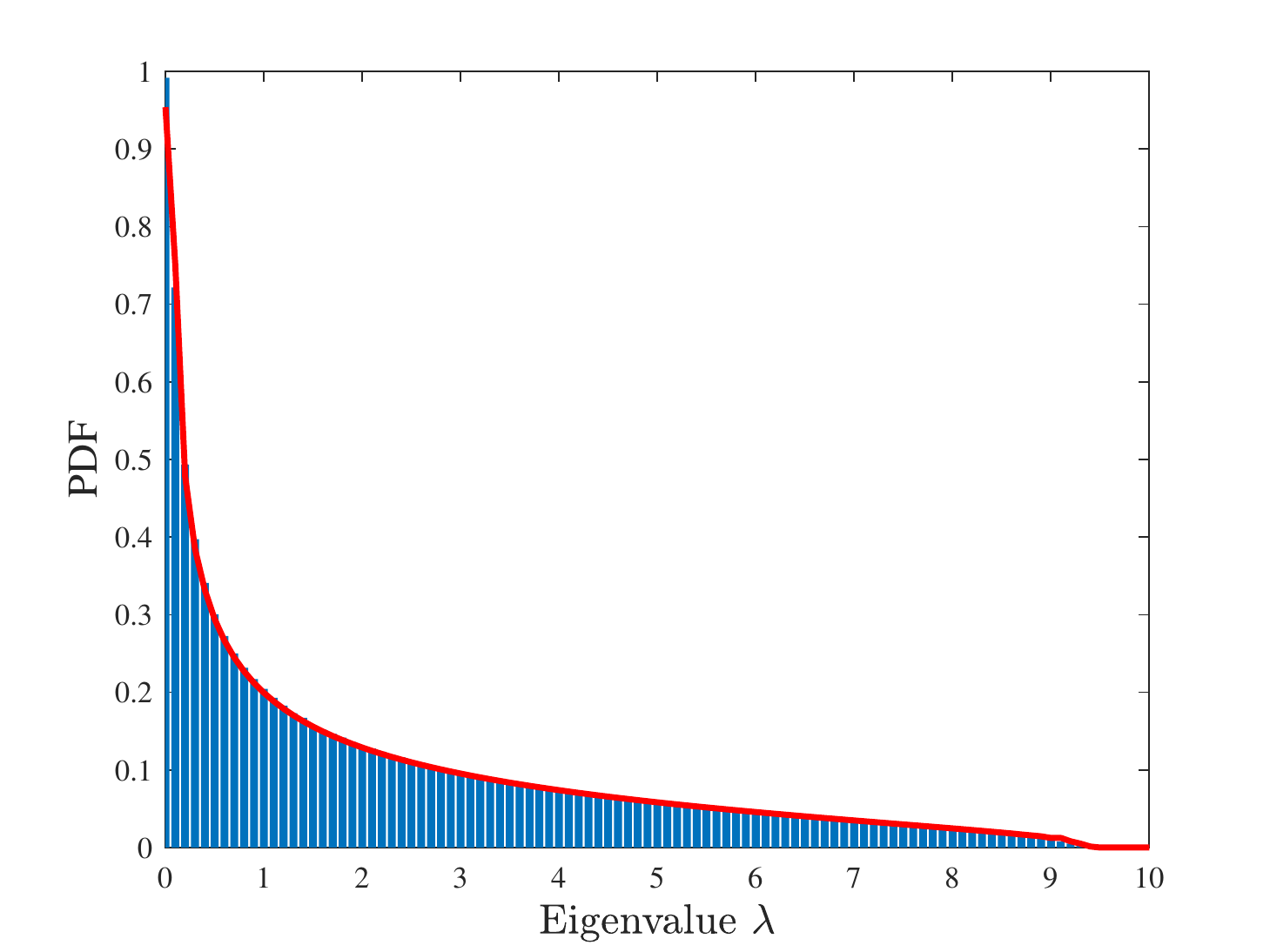}
	}
	\subfigure[$K = 4$]{\includegraphics[width=2.8in]{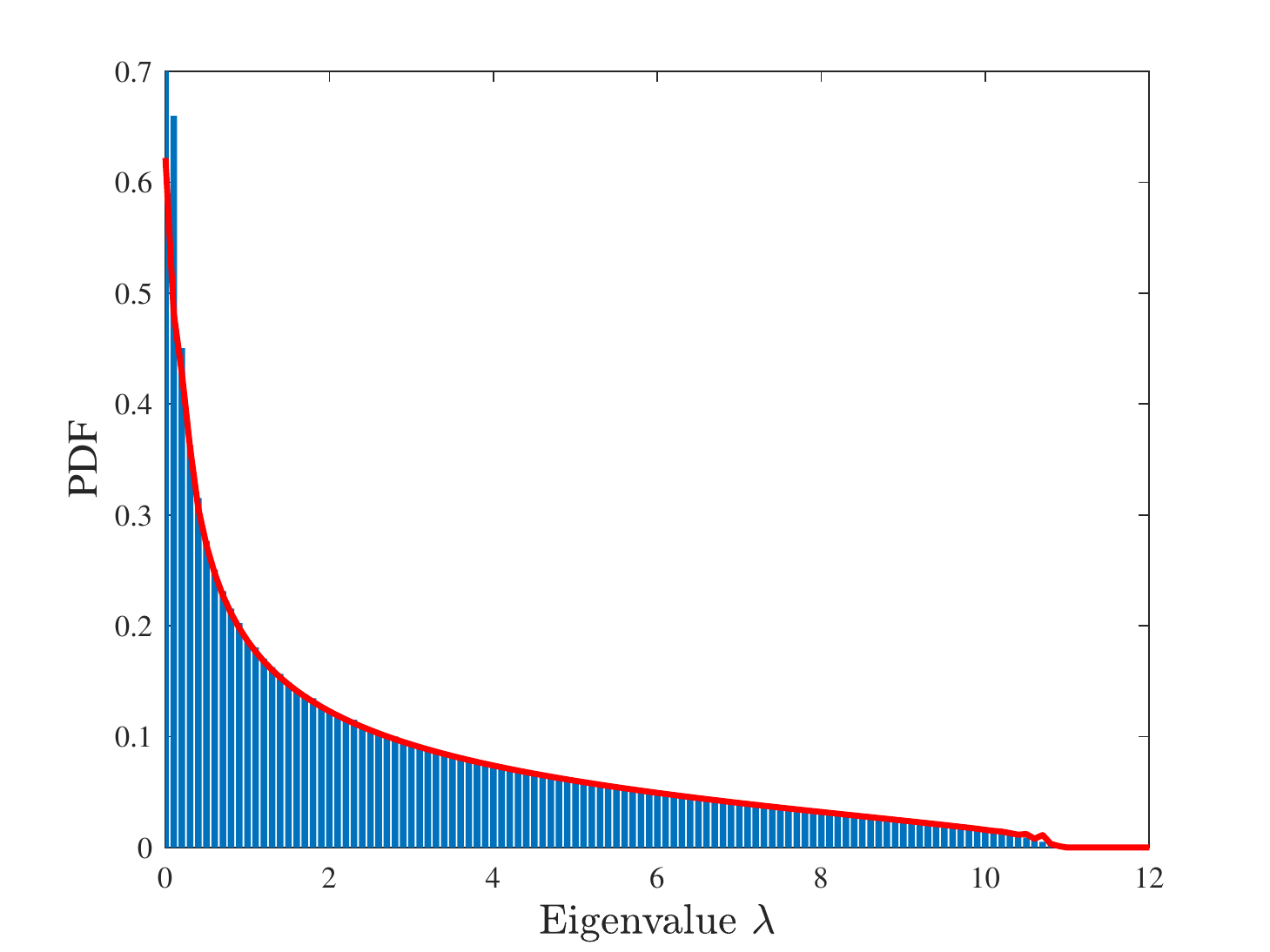}
	}
	\caption{Comparisons of empirical and asymptotic eigenvalue PDFs of the RIS-assisted MIMO channels $\mb{H}\mb{H}^\dagger$ with different numbers of RIS panels. The numbers of transmit and receive antennas are set to $T = R = 64$, and the number of reflecting elements of each RIS panel is set to $144$.}
	\label{figPDF}
\end{figure*}

In Fig.~\ref{figMI_SNR_T}, we investigate the impacts of the SNR, the number of antennas, and the Rician factors on the mutual information of the RIS-assisted MIMO channels. Equal number of antennas is set at the transmitter and the receiver, where $T = R = 4$ in Fig.~\ref{figMI_SNR_T}~(a) and $T = R = 8$ in Fig.~\ref{figMI_SNR_T}~(b), respectively. The MIMO communication is assisted by $K = 6$ RIS panels, and each RIS panel is composed of $16$ reflecting elements. Compared to the direct link $\mb{F}_0$, the relative channel gains $[\rho_1,\ldots,\rho_6]$ in (\ref{eqH}) corresponding to the reflected links are configured as $[0.9, 0.8, 0.7, 0.5, 0.3, 0.1]$. All the Rician factors are set equal as $\kappa = \kappa_{k}^{(F)} = \kappa_{k}^{(G)}$, where $\kappa$ is set to 1, 10, or 100. In presence of non-degenerate random scattering components $\widetilde{\mb{F}}_k$ in (\ref{eqFk}) and $\widetilde{\mb{G}}_k$ in (\ref{eqGk}), the RIS-assisted MIMO channels are full-rank, and the mutual information at large SNR linearly increases as $\min\{T,R\}/10\log_{10}(e)$ nats/s/Hz for every 1 dB SNR improvement, depicted as the dashed lines in Fig.~\ref{figMI_SNR_T}. However, as the Rician factor becomes large, although the mutual information has the same scaling law, it requires larger SNR levels to exhibit the linear improvement. This is illustrated in the insets of Fig.~\ref{figMI_SNR_T}. When the Rician factors are $\kappa = 1, 10$, and 100, the asymptotic mutual information has at least $5\%$ deviation from the high-SNR scaling law at SNRs 20.8 dB, 27.5 dB, and 34.2 dB when $T = R = 4$, and at SNRs 23.7 dB, 29.8 dB, and 36.1 dB when $T = R = 8$, respectively. This is due to the fact that as the Rician factor increases, the random scattering components to maintain the rank of the channel have less contributions to the overall MIMO channels.

\begin{figure*}[t!]
	\centering
	\subfigure[$T = R = 4$]{\includegraphics[width=0.6\columnwidth]{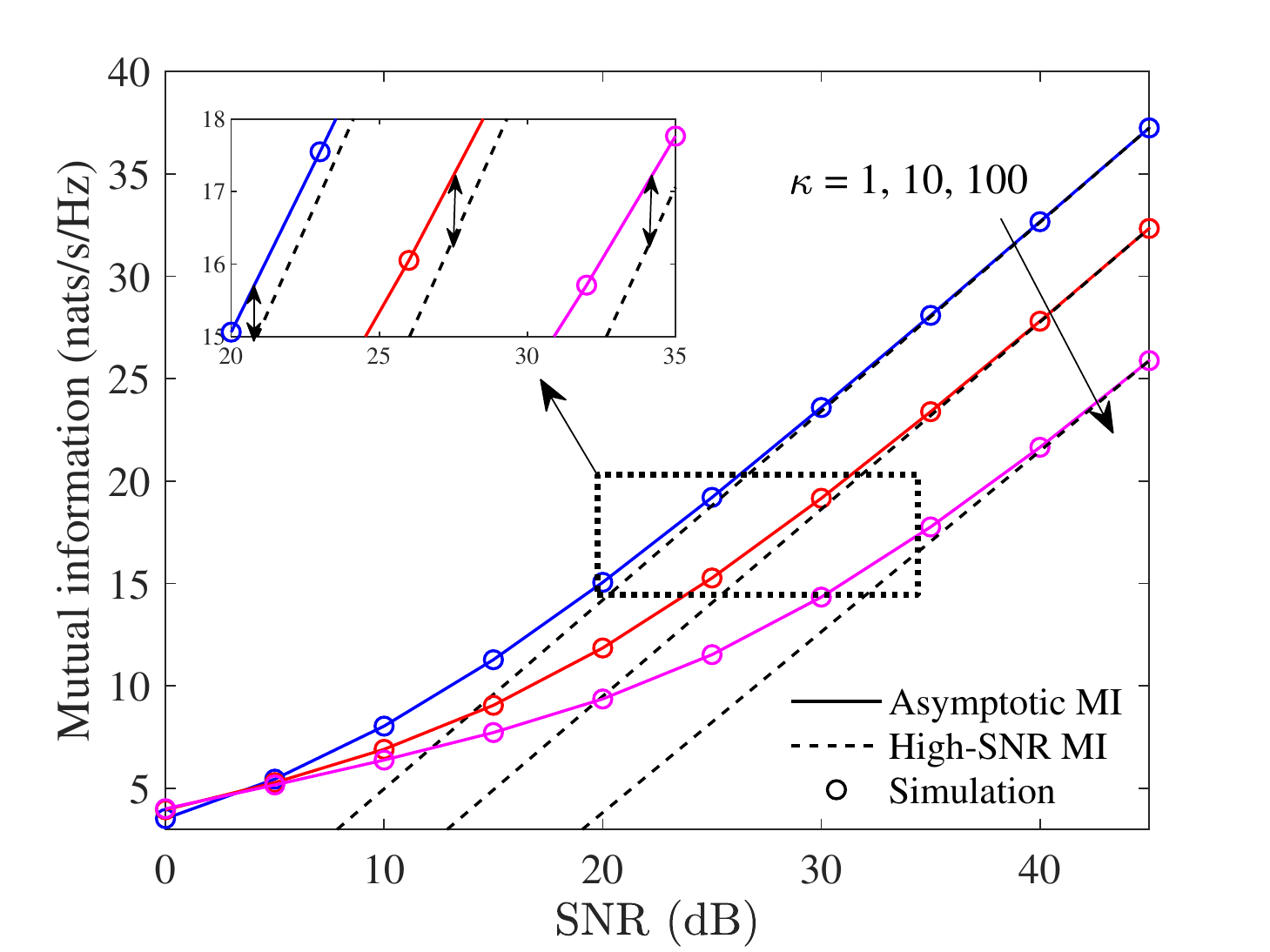}
	}\\
	\subfigure[$T = R = 8$]{\includegraphics[width=0.6\columnwidth]{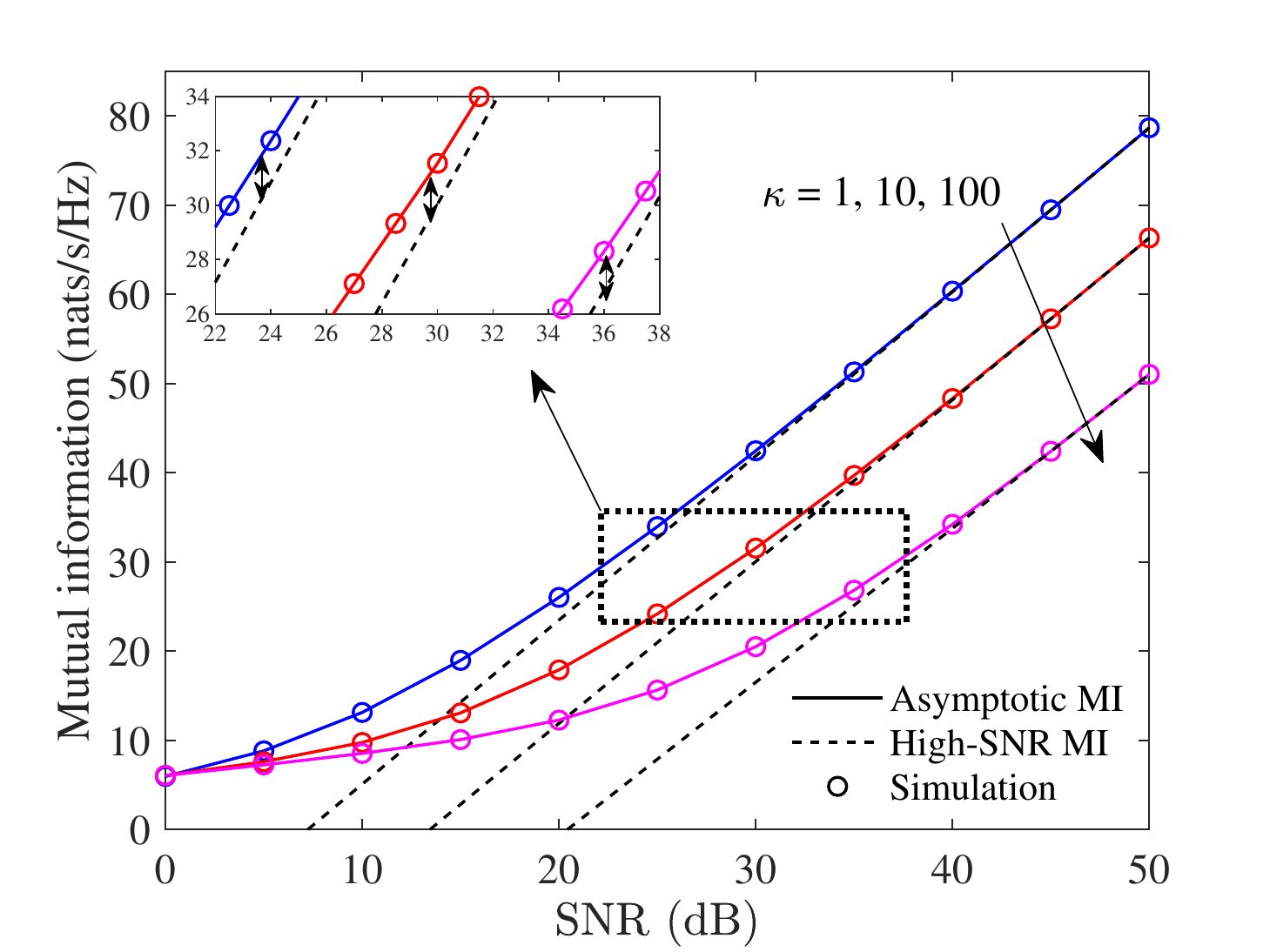}
	}
	\caption{Mutual information of RIS-assisted MIMO channels at varying SNR $\gamma$, when the number of antennas at the transceivers is $T = R = 4$ in (a) and $T = R = 8$ in (b), respectively. In each case, the Rician factor of the component channels is set equal to $\kappa = 1, 10$, or 100. There are $K = 6$ deployed RIS panels, each of which has 16 reflecting elements. Insets show the $5\%$ deviations of the asymptotic mutual information from the high-SNR scaling law.}
	\label{figMI_SNR_T}
\end{figure*}

To further investigate the impacts of the Rician factor on the mutual information of the MIMO channels, we plot Fig.~\ref{figMI_kappa_K} to show the mutual information as a function of $\kappa$, with the numbers of RIS panels $K$ set to 0, 1, 2, and 4, respectively. The number of transmit and receive antennas are set to $T = 16$ and $R = 8$, while the performance of the MIMO system is evaluated at SNR $\gamma = 10$ dB. It is observed that when $\kappa$ is less than 1, the mutual information can be improved as $\kappa$ increases, while it monotonically decreases for $\kappa>1$ in all the considered cases. When the number of RIS panels is larger, the mutual information degradation is less prominent as each RIS provides independent reflected link, which increases the richness of the MIMO channels.

\begin{figure}[t]
	\centerline{\includegraphics[width=0.6\columnwidth]{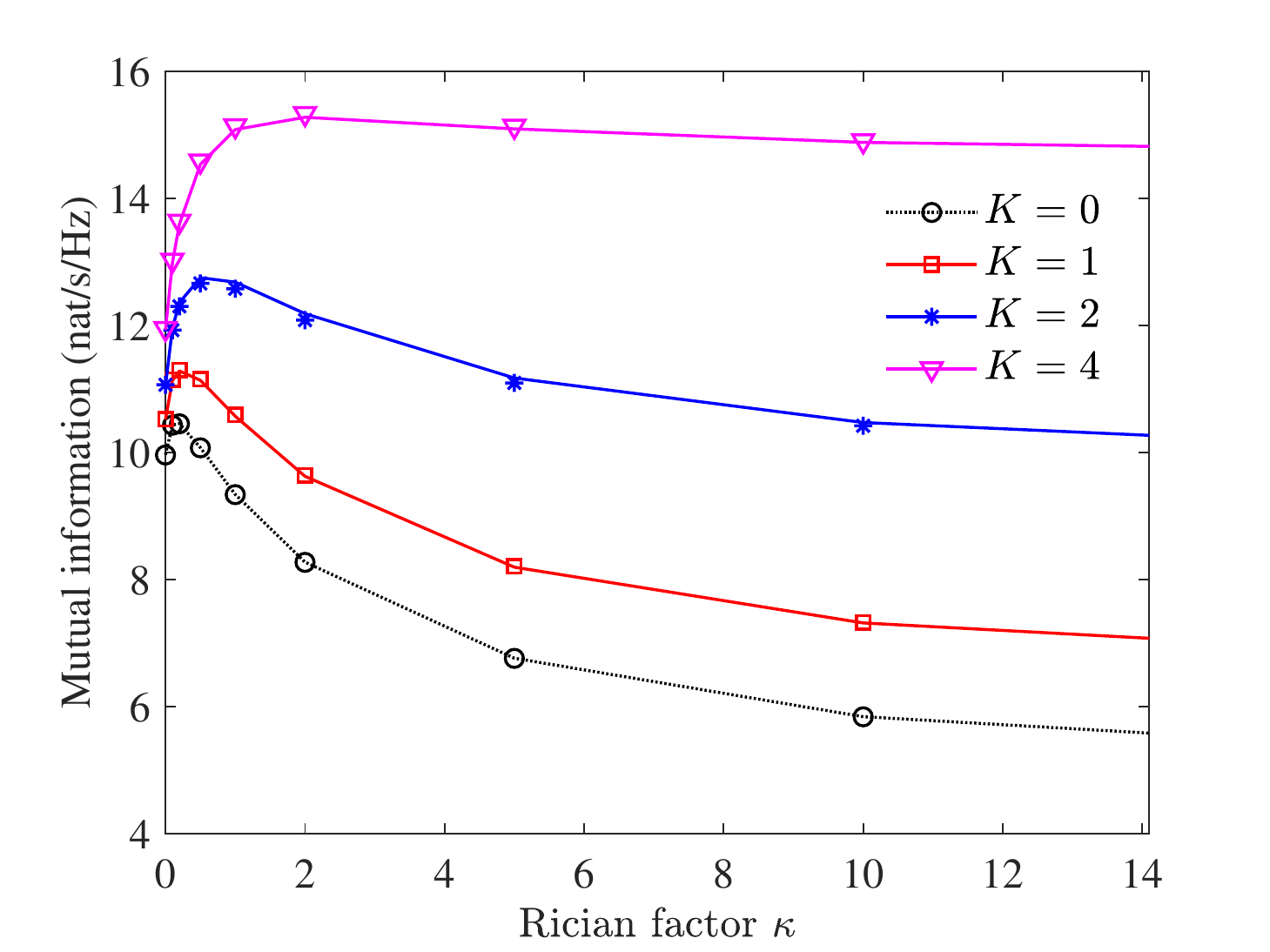}
	}
	\caption{Mutual information of RIS-assisted MIMO channel for varying Rician factor $\kappa$. The number of antennas at the transmitter and the receiver are $T = 16$ and $R = 8$, respectively, and each RIS panel has 8 reflecting elements. The SNR of the end-to-end channel is set to $\gamma = 10$ dB.}
	\label{figMI_kappa_K}
\end{figure}

%\begin{figure}[t]
%	\centerline{\includegraphics[width=0.6\columnwidth]{MI_SNR_K.eps}
%	}
%	\caption{.}
%	\label{figMI_SNR_K}
%\end{figure}

In Fig.~\ref{figMI_K_T}, the impact of the numbers of RIS panels is investigated in more details, when the mutual information is evaluated for different transmit antennas $T = 8$, 16, 32, and 64. The number of receive antennas is fixed to $R = 10$, and each RIS panel has 8 reflecting elements. In this simulation setting, we consider the urban canyon communication scenario as depicted in Fig.~\ref{figSystem}, where the specular components of $\{\mb{F}_k\}$ channels and of $\{\mb{G}_k\}$ channels have relatively small angular variations. That is, in (\ref{eqF_UPA}) and (\ref{eqG_UPA}), we assume that the departing angles $\left\{\theta_k^{(F)}, \phi_k^{(F)}\right\}_{0\le k\le K}$ of the transmitter UPA and the arriving angles $\left\{\varphi_k^{(G)}, \nu_k^{(G)}\right\}_{1\le k\le K}$ of the receiver UPA are uniformly and randomly generated in some fixed intervals having length $0.05\pi$. The departing angles $\left\{\theta_k^{(G)},\phi_k^{(G)}\right\}_{1\le k\le K}$ and the arriving angles $\left\{\varphi_k^{(F)},\nu_k^{(F)}\right\}_{1\le k\le K}$ of the RIS panels are randomly generated in some fixed intervals having length $0.1\pi$. As $K$ increases, Fig.~\ref{figMI_K_T} shows that the mutual information first improves at a larger rate between $0\le K\le 5$, and then becomes slower thereafter. This is due to the fact that the richness of the channels can be improved more efficiently when the number of reflected links is small. Since the angular ranges are restricted, the added RIS panels have similar reflected links that cannot provide additional richness. Therefore, it is less effective to deploy more RIS panels to improve the mutual information. Finally, as shown in Figs.~\ref{figMI_SNR_T}-\ref{figMI_K_T}, the mutual information calculated by (\ref{eqI_B}) via the Cauchy transform (\ref{eqGB}) achieves a good agreement with the simulation in all the considered simulation cases, and thus, can be applied to evaluate the performance of the RIS-assisted MIMO channels.

\begin{figure}[t]
	\centerline{\includegraphics[width=0.6\columnwidth]{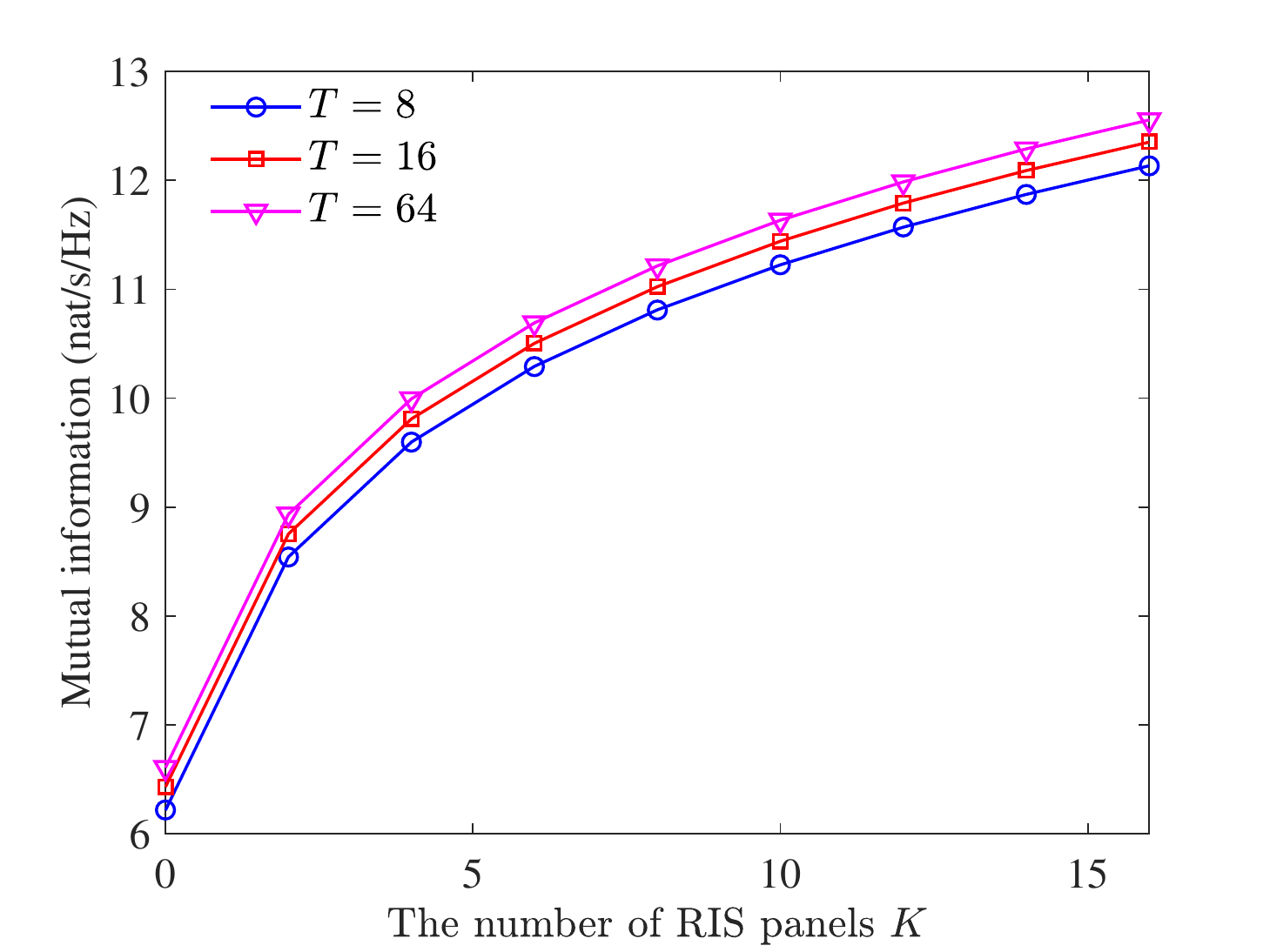}
	}
	\caption{Mutual information of RIS-assisted MIMO channel for varying numbers of RIS panels $K$. The number of receive antennas is $R = 8$, the number of elements in each RIS panel is 8, and the SNR of the channel is $\gamma = 10$ dB.}
	\label{figMI_K_T}
\end{figure}

\section{Conclusions}\label{secConclude}

This paper studies the information-theoretic data rate of the RIS-assisted MIMO systems, where multiple RIS panels are deployed to improve the scattering-limited MIMO channels. By using the operator-valued free probability theory, the Cauchy transform of the MIMO matrix is obtained using the general Rician MIMO model with Weichselberger's correlation structure. Based on this result, the asymptotic eigenvalue distribution of the channel matrix as well as the mutual information of the MIMO channel are calculated, which closely match the corresponding simulation results for practical system configurations. Numerical results show that the additional reflected links created by the RIS panels can increase the range of eigenvalues of the channel matrix, which can be leveraged to improve the amplitude of the eigen-channels. In the MIMO communications, the negative impact of a large Rician factor on the mutual information can be partly alleviated by deploying more RIS panels. However, the performance improvement of the multi-RIS deployment slows down when the added reflected links have similar arriving and departing angles.

\begin{appendices}
\section{Some Useful Matrix Inversion Identities}\label{appLemma}

For the sake of completeness, the following matrix inversion identities are summarized in Lemmas~\ref{lemmaSumInv}-\ref{lemmaBlockInv3}, which are repeatedly applied throughout this paper. For notational simplicity, in this appendix, we use {\em italic} bold symbols to define matrices, which are different from those used in the main sections.

\begin{lemma}{(Woodbury matrix inversion identity~\cite[Eq. (0.7.4.1)]{HornBook}.)}\label{lemmaSumInv}
	Let $\bm{A}$ denote a $m\times m$ invertible matrix, $\bm{D}$ denote a $k\times k$ matrix, $\bm{B}$ and $\bm{C}$ denote $m\times k$ and $k\times m$ matrices, respectively. Then the following identity holds
	\begin{align}
		\left(\bm{A}+\bm{B} \bm{D} \bm{C}\right)^{-1} = \bm{A}^{-1} - \bm{A}^{-1}\bm{B}\left(\bm{D}^{-1} + \bm{C} \bm{A}^{-1}\bm{B}\right)^{-1}\bm{C} \bm{A}^{-1}.
	\end{align}
\end{lemma}

\begin{lemma}{($2\times 2$ block matrix inversion identity~\cite[Eq. (0.7.3.1)]{HornBook}.)}\label{lemmaBlockInv2}
	Let $\bm{A}$, $\bm{B}$, $\bm{C}$, and $\bm{D}$ be defined as in Lemma~\ref{lemmaSumInv}, the inversion identity of the following $2\times 2$ block matrix holds
\begin{align}
	\begin{bmatrix}
		\bm{A} & \bm{B}\\
		\bm{C} & \bm{D}
	\end{bmatrix}^{-1} &= \begin{bmatrix}
	\bm{A}^{-1} + \bm{A}^{-1}\bm{B}(\bm{D}-\bm{C} \bm{A}^{-1} \bm{B})^{-1}\bm{C} \bm{A}^{-1} & -\bm{A}^{-1} \bm{B}(\bm{D}-\bm{C} \bm{A}^{-1}\bm{B})^{-1}\\
	-(\bm{D} - \bm{C} \bm{A}^{-1} \bm{B})^{-1} \bm{C} \bm{A}^{-1} & (\bm{D} - \bm{C} \bm{A}^{-1}\bm{B})^{-1}
\end{bmatrix}\nonumber\\
&= \begin{bmatrix}
	(\bm{A} - \bm{B} \bm{D}^{-1}\bm{C})^{-1} & -\bm{A}^{-1} \bm{B}(\bm{D}-\bm{C} \bm{A}^{-1} \bm{B})^{-1}\\
	-(\bm{D} - \bm{C} \bm{A}^{-1} \bm{B})^{-1} \bm{C} \bm{A}^{-1} & (\bm{D} - \bm{C} \bm{A}^{-1} \bm{B})^{-1}
\end{bmatrix},
\end{align}
where the second equality holds when $\bm{D}$ is also invertible.
\end{lemma}

\begin{lemma}{($3\times 3$ block matrix inversion identity.)}\label{lemmaBlockInv3}
	Let the matrices $\bm{E}$, $\bm{F}$, $\bm{G}$, $\bm{H}$, $\bm{J}$, $\bm{K}$, $\bm{L}$, $\bm{M}$, and $\bm{N}$ be the conformable partitions of the following $3\times 3$ block matrix $\bm{X}$
	\begin{align*}
		\bm{X} = \begin{bmatrix}
			\bm{E} & \bm{F} & \bm{G}\\
			\bm{H} & \bm{J} & \bm{K}\\
			\bm{L} & \bm{M} & \bm{N}
		\end{bmatrix}.
	\end{align*}
	When $\bm{E}$ is invertible, the inversion of $\bm{X}$ is given by 
	\begin{align*}
		\bm{X}^{-1} &= \begin{bmatrix}
		\bm{E}^{-1} + \bm{E}^{-1}(\bm{F} \bm{A}^{-1}\bm{H} + \bm{U} \bm{S}^{-1}\bm{V})\bm{E}^{-1} & -\bm{E}^{-1}(\bm{F} - \bm{U} \bm{S}^{-1} \bm{C})\bm{A}^{-1} & -\bm{E}^{-1} \bm{U} \bm{S}^{-1} \\
		-\bm{A}^{-1}(\bm{H} - \bm{B} \bm{S}^{-1} \bm{V})\bm{E}^{-1} & \bm{A}^{-1} + \bm{A}^{-1} \bm{B} \bm{S}^{-1} \bm{C} \bm{A}^{-1} & - \bm{A}^{-1} \bm{B} \bm{S}^{-1}\\
		-\bm{S}^{-1} \bm{V} \bm{E}^{-1} & -\bm{S}^{-1} \bm{C} \bm{A}^{-1} & \bm{S}^{-1}
	\end{bmatrix},
	\end{align*}
	where 
	\begin{align}
		\bm{A} &= \bm{J} - \bm{H} \bm{E}^{-1}\bm{F},\quad	\bm{B} = \bm{K} - \bm{H} \bm{E}^{-1}\bm{G},\quad \bm{C} = \bm{M} - \bm{L} \bm{E}^{-1}\bm{F},\quad	\bm{D} = \bm{N} - \bm{L} \bm{E}^{-1}\bm{G},\label{eqABCD}\\
		\bm{U} &= \bm{G} - \bm{F} \bm{A}^{-1}\mb{B},\quad	\bm{V} = \bm{L} - \bm{C} \bm{A}^{-1}\bm{H},\label{eqUV}\\
		\bm{S} &= \bm{D} - \bm{C} \bm{A}^{-1}\bm{B}.\label{eqS}
	\end{align}
\end{lemma}
\begin{proof}
	Apply Lemma~\ref{lemmaBlockInv2} to the inversion of $\bm{X}$ that is partitioned into a $2\times 2$ block matrix as
	\begin{align}
		\bm{X}^{-1} = \left[\begin{Array}{c:cc}
			\bm{E} & \bm{F} & \bm{G}\\\hdashline
			\bm{H} & \bm{J} & \bm{K}\\
			\bm{L} & \bm{M} & \bm{N}
		\end{Array}\right]^{-1} = \begin{bmatrix}
		\bm{P} & \bm{Q}\\
\bm{R} & \bm{Z}^{-1}
	\end{bmatrix},\label{eqXinv}
	\end{align}
where the matrix blocks $\bm{P}$, $\bm{Q}$, and $\bm{R}$ are given by
\begin{align}
	\bm{P} &= \bm{E}^{-1} + \bm{E}^{-1}\begin{bmatrix}
		\bm{F} & \bm{G}
	\end{bmatrix} \bm{Z}^{-1} \begin{bmatrix}
		\bm{H}\\ \bm{L}
	\end{bmatrix} \bm{E}^{-1},\label{eqP}\\
	\bm{Q} &= -\bm{E}^{-1} \begin{bmatrix}
		\bm{F} & \bm{G}
	\end{bmatrix} \bm{Z}^{-1},\label{eqQ}\\
	\bm{R} &= -\bm{Z}^{-1} \begin{bmatrix}
		\bm{H}\\ \bm{L}
	\end{bmatrix} \bm{E}^{-1},\label{eqRmat}
\end{align}
and the matrix $\bm{Z}$ is a $2\times 2$ block matrix such that
\begin{align}
	\bm{Z} = \begin{bmatrix}
		\bm{J} & \bm{K}\\
		\bm{M} & \bm{N}
		\end{bmatrix} - \begin{bmatrix}
		\bm{H}\\ \bm{L}
	\end{bmatrix} \bm{E}^{-1} \begin{bmatrix}
	\bm{F} & \bm{G}
\end{bmatrix} = \begin{bmatrix}
\bm{A} & \bm{B}\\
\bm{C} & \bm{D}
\end{bmatrix},
\end{align}
where $\bm{A}$, $\bm{B}$, $\bm{C}$, and $\bm{D}$ are given in (\ref{eqABCD}).

Applying again Lemma~\ref{lemmaBlockInv2} to the inversion of $\mb{Z}$, we obtain
\begin{align}
	\bm{Z}^{-1} = \begin{bmatrix}
		\bm{A}^{-1} + \bm{A}^{-1}\bm{B}\bm{S}^{-1}\bm{C}\bm{A}^{-1} & -\bm{A}^{-1}\bm{B}\bm{S}^{-1}\\
		-\bm{S}^{-1}\bm{C}\bm{A}^{-1} & \bm{S}^{-1}
	\end{bmatrix},\label{eqZinv}
\end{align}
where $\bm{S}$ is given in (\ref{eqS}). Substituting (\ref{eqZinv}) into (\ref{eqP}), $\bm{P}$ can be rewritten as
\begin{align}
	\bm{P} &= \bm{E}^{-1} + \bm{E}^{-1} \begin{bmatrix}
		\bm{F}\bm{A}^{-1} - \bm{U}\bm{S}^{-1}\bm{C}\bm{A}^{-1} & \bm{U}\bm{S}^{-1}
	\end{bmatrix}\begin{bmatrix}
	\bm{H}\\\bm{L}
\end{bmatrix} \bm{E}^{-1}\nonumber\\
	&= \bm{E}^{-1} + \bm{E}^{-1} \left( \bm{F}\bm{A}^{-1}\bm{H} + \bm{U}\bm{S}^{-1}\bm{V}  \right) \bm{E}^{-1},
\end{align} 
where $\bm{U}$ and $\bm{V}$ are defined in (\ref{eqUV}). Similarly, $\bm{Q}$ and $\bm{R}$ can be obtained as
\begin{align}
	\bm{Q} &= \begin{bmatrix}
		-\bm{E}^{-1}\left(\bm{F} -(\bm{G} - \bm{F}\bm{A}^{-1}\bm{B})\bm{S}^{-1}\bm{C} \right)\bm{A}^{-1} & -\bm{E}^{-1}\left(\bm{G} - \bm{F}\bm{A}^{-1}\bm{B}\right)\bm{S}^{-1}
	\end{bmatrix}\nonumber\\
		&= \begin{bmatrix}
			-\bm{E}^{-1}\left(\bm{F} -\bm{U}\bm{S}^{-1}\bm{C} \right)\bm{A}^{-1} & -\bm{E}^{-1}\bm{U}\bm{S}^{-1}
		\end{bmatrix},\\
	\bm{R} &= \begin{bmatrix}
		-\bm{A}^{-1}\bm{H}\bm{E}^{-1} + \bm{A}^{-1}\bm{B}\bm{S}^{-1}(\bm{L} - \bm{C}\bm{A}^{-1}\bm{H})\bm{E}^{-1} \\ -\bm{S}^{-1}(\bm{L} - \bm{C}\bm{A}^{-1}\bm{H})\bm{E}^{-1}
	\end{bmatrix} = \begin{bmatrix}
	-\bm{A}^{-1}(\bm{H} - \bm{B}\bm{S}^{-1}\bm{V})\bm{E}^{-1} \\ -\bm{S}^{-1}\bm{V}\bm{E}^{-1}
\end{bmatrix}.\label{eqRmat2}
\end{align}
Finally, substituting (\ref{eqZinv})-(\ref{eqRmat2}) into (\ref{eqXinv}) completes the proof of Lemma~\ref{lemmaBlockInv3}.
\end{proof}

\section{Proof of Proposition~\ref{prop_freeness}}\label{appx_freeness}
	A random variable $\widetilde{\mb{L}}\in\mc{M}$ is said to be $\mc{D}$-valued semicircular if the free cumulant 
\begin{align}
	\kappa_m^{\mc{D}}(\widetilde{\mb{L}} b_1, \widetilde{\mb{L}} b_2, \ldots, \widetilde{\mb{L}} b_{m-1}, \widetilde{\mb{L}}) = 0,
\end{align}
for all $n\neq 2$, and all $b_1,\ldots,b_{n-1}\in\mc{D}$. The free cumulant $\kappa_m^{\mc{D}}$ is a mapping from $\mc{M}^m$ to $\mc{D}$ and we refer the reader to~\cite{MingoSpeicherBook} for detailed explanations on this topic. The proof is followed by expanding $\widetilde{\mb{L}}$ into a sum of $n\times n$ matrices, such that
\begin{align}
	\widetilde{\mb{L}} = \sum_{k = 0}^K \widetilde{\mb{L}}_k^{(F)} + \sum_{k = 1}^K \widetilde{\mb{L}}_k^{(G)},
\end{align}
where the matrices $\widetilde{\mb{L}}_k^{(F)}$ and $\widetilde{\mb{L}}_k^{(G)}$ are given by
\begin{align}
	\widetilde{\mb{L}}_k^{(F)} &= \left[\begin{array}{c:c:c:c}
		& & & \mb{0}_{R\times L} \\\hdashline
		& & \widehat{{\mb{F}}}_k & \\\hdashline
		& \widehat{{\mb{F}}}_k^\dagger & & \\\hdashline
		\mb{0}_{L\times R} & & &
	\end{array}\right],\label{eqLkF}\\
	\widetilde{\mb{L}}_k^{(G)} &= \left[\begin{array}{c:c:c:c}
		& & & \widehat{{\mb{G}}}_k \\\hdashline
		& & \mb{0}_{L\times T} & \\\hdashline
		& \mb{0}_{T\times L} & & \\\hdashline
		\widehat{{\mb{G}}}_k^\dagger & & &
	\end{array}\right],\label{eqLkG}
\end{align}
where $\widehat{{\mb{F}}}_k$ and $\widehat{{\mb{G}}}_k$ are $L\times T$ and $R\times L$ matrices, respectively, and are given by
\begin{align}
	\widehat{{\mb{F}}}_k &= \begin{bmatrix}
		\mb{0}_{T\times L_0} & \ldots & \widetilde{\mb{F}}_k^\dagger & \ldots &\mb{0}_{T\times L_K}
	\end{bmatrix}^\dagger, \quad 0\le k\le K,\label{eqFhat}\\
	\widehat{{\mb{G}}}_k &= \begin{bmatrix}
		\mb{0}_{R\times R} & \mb{0}_{R\times L_1} & \ldots & \sqrt{\rho_k} \widetilde{\mb{G}}_k & \ldots &\mb{0}_{R\times L_K}
	\end{bmatrix},\quad 1\le k\le K.\label{eqGhat}
\end{align}

Recalling the definitions of $\widetilde{\mb{F}}_k$ and $\widetilde{\mb{G}}_k$ in (\ref{eqFk}) and (\ref{eqGk}), we have
\begin{align}
	\widetilde{\mb{L}}_k^{(F)} &= \bm{\mc{A}}_k^{(F)} \widetilde{\bm{\mc{X}}}_k\bm{\mc{A}}_k^{(F)\dagger},\\
	\widetilde{\mb{L}}_k^{(G)} &= \bm{\mc{A}}_k^{(G)} \widetilde{\bm{\mc{Y}}}_k\bm{\mc{A}}_k^{(G)\dagger},
\end{align}
where the matrix $\widetilde{\bm{\mc{X}}}_k$ has the same structure as the block matrix $\widetilde{\mb{L}}_k^{(F)}$ in (\ref{eqLkF}) while replacing $\widetilde{\mb{F}}_k$ in (\ref{eqFhat}) with $\widetilde{\mb{X}}_k = \mb{M}_k\odot\mb{X}_k$, and $\widetilde{\bm{\mc{Y}}}_k$ has the same structure as the block matrix $\widetilde{\mb{L}}_k^{(G)}$ in (\ref{eqLkG}) while replacing $\widetilde{\mb{G}}_k$ in (\ref{eqGhat}) with $\widetilde{\mb{Y}}_k = \frac{1}{\sqrt{r_k}}\mb{N}_k\odot\mb{Y}_k$. The $n\times n$ matrices $\bm{\mc{A}}_k^{(F)}$ and $\bm{\mc{A}}_k^{(G)}$ are given by
\begin{align}
	\bm{\mc{A}}_k^{(F)} &= \begin{bmatrix}
		\widehat{\mb{U}}_k & \mb{0}_{(R+L)\times(T+L)} \\ \mb{0}_{(T+L)\times (R+L)} & \widehat{\mb{V}}_k
	\end{bmatrix},\\
	\bm{\mc{A}}_k^{(G)} &= \begin{bmatrix}
		\widehat{\mb{W}}_k & \mb{0}_{(R+L)\times(T+L)} \\ \mb{0}_{(T+L)\times(R+L)} & \widehat{\mb{S}}_k
	\end{bmatrix},
\end{align}
where $\widehat{\mb{U}}_k$, $\widehat{\mb{V}}_k$, $\widehat{\mb{W}}_k$, $\widehat{\mb{S}}_k$ are deterministic diagonal block matrices and are given by
\begin{align}
	\widehat{\mb{U}}_k &= \mathrm{blkdiag}(\mb{0}_R, \mb{0}_{L_0},\ldots,\mb{U}_k,\ldots,\mb{0}_{L_K}),\\
	\widehat{\mb{V}}_k &= \mathrm{blkdiag}(\mb{V}_k,\mb{0}_L),\\
	\widehat{\mb{W}}_k &= \mathrm{blkdiag}(\mb{W}_k,\mb{0}_{L}),\\
	\widehat{\mb{S}}_k &= \mathrm{blkdiag}(\mb{0}_T, \mb{0}_{L_0},\ldots,\mb{S}_k,\ldots,\mb{0}_{L_K}).
\end{align}
Since $\{\widetilde{\bm{\mc{X}}}_k\}_{0\le k\le K}$, $\{\widetilde{\bm{\mc{Y}}}_k\}_{1\le k\le K}$ are Wigner matrices and independent from each other, they are semicircular and free over the sub-algebra $\mc{D}_n\subset\mc{M}$ of $n\times n$ diagonal matrices. Then, following the same arguments as in~\cite[Appendix B]{LuTIT2016}, $\{\widetilde{\mb{L}}_k^{(F)}\}_{0\le k\le K}$ and $\{\widetilde{\mb{L}}_k^{(G)}\}_{1\le k\le K}$ are semicircular and free over sub-algebra of block diagonal matrices $\mc{D}$. Therefore, the sum of $\widetilde{\mb{L}}_k^{(F)}$ and $\widetilde{\mb{L}}_k^{(G)}$ is also semicircular over $\mc{D}$ and is free from any deterministic matrix from $\mc{M}$.

\section{Proof of Proposition~\ref{prop_cauchyB}}\label{appx_cauchyB}
Since $\widetilde{\mb{L}}$ is an operator-valued semicircular variable over $\mc{D}$ and $\widetilde{\mb{L}}$ are free from $\overline{\mb{L}}$ over $\mc{D}$, the limiting spectral distribution of $\mb{L}$ is a free additive convolution of the limiting spectral distributions of $\widetilde{\mb{L}}$ and $\overline{\mb{L}}$. Specifically, the operator-valued Cauchy transform $\mc{G}_{\mb{L}}^{\mc{D}}$ can be calculated via the subordination formula~(\ref{eqGL_sub}). Recall that the $R$-transform $\mc{R}_{\widetilde{\mb{L}}}^{\mc{D}}\left(\cdot\right)$ is the free cumulant generating function of $\widetilde{\mb{L}}$ with the following formal power series expansion:
\begin{align}
	\mc{R}_{\widetilde{\mb{L}}}^{\mc{D}}\left(\mb{K}\right) = \kappa_1^{\mc{D}}(\widetilde{\mb{K}}) + \kappa_2^{\mc{D}}(\widetilde{\mb{L}}\mb{K},\widetilde{\mb{L}}) + \kappa_3^{\mc{D}}(\widetilde{\mb{L}}\mb{K},\widetilde{\mb{L}}\mb{K},\widetilde{\mb{L}}) + \cdots,\label{eqRBL}
\end{align}
where $\kappa_i^{\mc{D}}$ denotes the $i$-th free cumulant of $\widetilde{\mb{L}}$ over $\mc{D}$. In addition, since $\widetilde{\mb{L}}$ is semicircular over $\mc{D}$, all its cumulants in (\ref{eqRBL}) except $\kappa_2^{\mc{D}}$ are zero. Therefore, the $R$-transform $\mc{R}_{\widetilde{\mb{L}}}^{\mc{D}}\left(\mb{K}\right)$ reduces to the covariance function of $\widetilde{\mb{L}}$ over $\mc{D}$ parameterized by $\mb{K}$, i.e.,
\begin{align}
	\mc{R}_{\widetilde{\mb{L}}}^{\mc{D}}(\mb{K}) &= \mbb{E}_{\mc{D}}\left[\widetilde{\mb{L}}\mb{K}\widetilde{\mb{L}}\right]\nonumber\\
	&= \begin{bmatrix}
		\sum_{k=1}^K \widetilde{\eta}_k(\mb{C}_k) & & & \\
		& \widetilde{\bm{\zeta}}(\widetilde{\mb{D}}) & & \\
		& & \sum_{k=0}^K \zeta_k(\mb{D}_k) & \\
		& & & \bm{\eta}(\widetilde{\mb{C}})\\
	\end{bmatrix},\label{eqR}
\end{align}
where $\widetilde{\bm{\zeta}}(\widetilde{\mb{D}}) = \mathrm{blkdiag}\left\{\widetilde{\zeta}_0(\widetilde{\mb{D}}),\ldots,\widetilde{\zeta}_K(\widetilde{\mb{D}})\right\}$ and $\bm{\eta}(\widetilde{\mb{C}}) = \mathrm{blkdiag}\left\{\mb{0}_R, \eta_1(\widetilde{\mb{C}}),\ldots,\eta_K(\widetilde{\mb{C}})\right\}$.

Since $\mc{G}_{\mb{L}}^{\mc{D}}(\mb{\Lambda}(z))\in\mc{D}$, by same matrix partitioning as in (\ref{eqK}), $\mc{G}_{\mb{L}}^{\mc{D}}(\mb{\Lambda}(z))$ is partitioned into 
\begin{align}
	\mc{G}_{\mb{L}}^{\mc{D}}(\mb{\Lambda}(z)) = \mathrm{blkdiag}\left\{\mc{G}_{\widetilde{\mb{C}}}(z), \bm{\mc{G}}_{{\mb{D}}}(z),{\mc{G}}_{\widetilde{\mb{D}}}(z),\bm{\mc{G}}_{{\mb{C}}}(z)\right\},\label{eqGBL_partition}
\end{align}
where $\bm{\mc{G}}_{{\mb{D}}}(z) = \mathrm{blkdiag}\left\{{\mc{G}}_{{\mb{D}_0}}(z),\ldots,{\mc{G}}_{{\mb{D}_K}}(z)\right\}$ and $\bm{\mc{G}}_{{\mb{C}}}(z) = \mathrm{blkdiag}\left\{\mb{0}_R, {\mc{G}}_{{\mb{C}_1}}(z),\ldots,{\mc{G}}_{{\mb{C}_K}}(z)\right\}$. Note that the upper-left block $\left\{\mc{G}_{\mb{L}}^{\mc{D}}(\mb{\Lambda}(z))\right\}^{(1,1)} = \mc{G}_{\widetilde{\mb{C}}}(z)$, which is then used to compute $\mc{G}_{\mb{B}}(z) = \frac{1}{R}\mathrm{Tr}(\mc{G}_{\widetilde{\mb{C}}}(z))$. 

By replacing $\mb{K}$ in (\ref{eqR}) with $\mc{G}_{\mb{L}}^{\mc{D}}(\mb{\Lambda}(z))$ in (\ref{eqGBL_partition}), and substituting $\overline{\mb{L}}$ and $\mc{R}_{\widetilde{\mb{L}}}^{\mc{D}}$ with (\ref{eqBL_bar}) and (\ref{eqR}), respectively, we obtain $\mc{G}_{\mb{L}}^{\mc{D}}(\mb{\Lambda}(z))$ as
\begin{align}
	\mc{G}_{\mb{L}}^{\mc{D}}(\mb{\Lambda}(z)) &= \begin{bmatrix}
		\mc{G}_{\widetilde{\mb{C}}}(z) & & & \\
		& \bm{\mc{G}}_{{\mb{D}}}(z) & & \\
		& & {\mc{G}}_{\widetilde{\mb{D}}}(z) & \\
		& & & \bm{\mc{G}}_{{\mb{C}}}(z)
	\end{bmatrix} = \mbb{E}_{\mc{D}} \begin{pmatrix}
		\widetilde{\mb{\Psi}}(z) & \mb{0} & \mb{0} & -\overline{\mb{G}} \\
		\mb{0} & \widetilde{\mb{\Phi}}(z) & -\overline{\mb{F}} & \mb{I}_{L}\\
		\mb{0} & -\overline{\mb{F}}^\dagger & \mb{\Phi}(z) & \mb{0} \\
		-\overline{\mb{G}}^\dagger & \mb{I}_{L} & \mb{0} & \mb{\Psi}(z)\\		
	\end{pmatrix}^{-1},\label{eqGBL}
\end{align}
where $\widetilde{\mb{\Psi}}(z)$, $\mb{\Psi}(z)$, $\widetilde{\mb{\Phi}}(z)$, and $\mb{\Phi}(z)$ are given in (\ref{eqPsit})-(\ref{eqPhi}). By invoking Lemma~\ref{lemmaBlockInv2} to the RHS of (\ref{eqGBL}) and taking expectation over $\mc{D}$,
the matrix-valued function $\mc{G}_{\widetilde{\mb{C}}}(z) = \mb{A}_1^{-1}$, and $\mc{G}_{\mb{D}}(z)$, $\mc{G}_{\widetilde{\mb{D}}}(z)$, $\mc{G}_{\mb{C}}(z)$ are the diagonal blocks of the matrix $\mb{A}_2^{-1}$, where $\mb{A}_1$ and $\mb{A}_2$ are given by
\begin{align}
	\mb{A}_1 &= \widetilde{\mb{\Psi}}(z) - \begin{bmatrix}
		\mb{0} & \mb{0} & \overline{\mb{G}}
	\end{bmatrix} \begin{bmatrix}
		\widetilde{\mb{\Phi}}(z) & -\overline{\mb{F}} & \mb{I}_{L}\\
		-\overline{\mb{F}}^\dagger & \mb{\Phi}(z) & \mb{0} \\
		\mb{I}_{L} & \mb{0} & \mb{\Psi}(z)
	\end{bmatrix}^{-1} \begin{bmatrix}
	\mb{0} \\ \mb{0} \\ \overline{\mb{G}}^\dagger
\end{bmatrix} ,\label{eqA1}\\
	\mb{A}_2 &= \begin{bmatrix}
		\widetilde{\mb{\Phi}}(z) & -\overline{\mb{F}} & \mb{I}_{L}\\
		-\overline{\mb{F}}^\dagger & \mb{\Phi}(z) & \mb{0} \\
		\mb{I}_{L} & \mb{0} & \mb{\Psi}(z) 
	\end{bmatrix} - \begin{bmatrix}
	\mb{0} \\ \mb{0} \\ \overline{\mb{G}}^\dagger
\end{bmatrix} \widetilde{\mb{\Psi}}(z)^{-1}  \begin{bmatrix}
\mb{0} & \mb{0} & \overline{\mb{G}}
\end{bmatrix}.
\end{align}
Applying Lemma~\ref{lemmaBlockInv3}, the RHS of (\ref{eqA1}) can be further derived as
\begin{align}
	\mb{A}_1 = \widetilde{\mb{\Psi}}(z) - \overline{\mb{G}} \bm{S}^{-1} \overline{\mb{G}}^\dagger,\label{eqA1_2}
\end{align}
where $\bm{S}=\mb{\Xi}(z)$ and is calculated in (\ref{eqS}) as
\begin{align}
	\bm{S} &= \mb{\Xi}(z) = \mb{\Psi}(z) - \widetilde{\mb{\Phi}}(z)^{-1} - \widetilde{\mb{\Phi}}(z)^{-1}\overline{\mb{F}}\left(\mb{\Phi}(z) - \overline{\mb{F}}^\dagger\widetilde{\mb{\Phi}}(z)^{-1}\overline{\mb{F}}\right)^{-1}\overline{\mb{F}}^\dagger\widetilde{\mb{\Phi}}(z)^{-1}\nonumber\\
	&= \mb{\Psi}(z) - \left(\widetilde{\mb{\Phi}}(z) - \overline{\mb{F}}\mb{\Phi}(z)^{-1}\overline{\mb{F}}^\dagger\right)^{-1}.\label{eqXi}
\end{align} 
The second equality of (\ref{eqXi}) is obtained by applying Lemma~\ref{lemmaSumInv}. Then,  (\ref{eqGCt}) is established by combining (\ref{eqA1_2}) and (\ref{eqXi}). 

The inverse of $\mb{A}_2$ can be explicitly calculated via Lemma~\ref{lemmaBlockInv3}, where $\bm{E} = \widetilde{\mb{\Phi}}(z)$, $\bm{F} = \bm{H}^\dagger = -\overline{\mb{F}}$, $\bm{G} = \bm{L} = \mb{I}_L$, $\bm{J} = \mb{\Phi}(z)$, $\bm{K} = \bm{M} = \mb{0}$, and $\bm{N} = \mb{\Psi}(z) - \overline{\mb{G}}^\dagger \widetilde{\mb{\Psi}}(z)^{-1}\overline{\mb{G}}$. We further let $\bm{T} = \mb{\Phi}(z) - \overline{\mb{F}}^\dagger \widetilde{\mb{\Phi}}(z)^{-1} \overline{\mb{F}}$ and $\widetilde{\bm{T}} = \widetilde{\mb{\Phi}}(z) - \overline{\mb{F}} \mb{\Phi}(z)^{-1} \overline{\mb{F}}^\dagger$. Then, the matrix-valued functions $\mc{G}_{\mb{D}}(z)$, $\mc{G}_{\widetilde{\mb{D}}}(z)$, and $\mc{G}_{\mb{C}}(z)$, being the diagonal blocks of $\mb{A}_2^{-1}$, are given by
\begin{align}
	\mc{G}_{\mb{D}}(z) &= \widetilde{\mb{\Phi}}(z)^{-1} + \widetilde{\mb{\Phi}}(z)^{-1}\overline{\mb{F}} \bm{T}^{-1} \overline{\mb{F}}^\dagger\widetilde{\mb{\Phi}}(z)^{-1} + \widetilde{\bm{T}}^{-1} \left( \bm{N} - \widetilde{\bm{T}}^{-1}  \right)^{-1} \widetilde{\bm{T}}^{-1},\label{eqGD}\\
	\mc{G}_{\widetilde{\mb{D}}}(z) &= \bm{T}^{-1} + \bm{T}^{-1} \overline{\mb{F}}^\dagger\widetilde{\mb{\Phi}}(z)^{-1} \left(\bm{N} - \widetilde{\bm{T}}^{-1}\right)^{-1} \widetilde{\mb{\Phi}}(z)^{-1}\overline{\mb{F}}\ \bm{T}^{-1}, \label{eqGDt_2}\\
	\mc{G}_{\mb{C}}(z) &= \left(\bm{N} - \left(\widetilde{\mb{\Phi}}(z)^{-1} + \widetilde{\mb{\Phi}}(z)^{-1}\overline{\mb{F}} \bm{T}^{-1} \overline{\mb{F}}^\dagger\widetilde{\mb{\Phi}}(z)^{-1}\right)\right)^{-1}.\label{eqGC}
\end{align}
Finally, applying Lemma~\ref{lemmaSumInv} to (\ref{eqGD})-(\ref{eqGC}), we obtain $\mc{G}_{\mb{D}}(z)$, $\mc{G}_{\widetilde{\mb{D}}}(z)$, and $\mc{G}_{\mb{C}}(z)$ as in (\ref{eqGCk})-(\ref{eqGDk}).
\end{appendices}

\ifCLASSOPTIONcaptionsoff
  \newpage
\fi

\end{document}